\documentclass[journal,comsoc]{IEEEtran}
\usepackage{caption}
\usepackage{subcaption}
\usepackage{lipsum}
\usepackage{url}

\usepackage{color}

\usepackage{graphics}
\usepackage{amsmath,amsfonts, amssymb}
\usepackage{mathtools}
\usepackage{epstopdf}
\usepackage{amsthm}
\usepackage{bbm}
\usepackage{algorithm}
\usepackage[noend]{algpseudocode}
\usepackage{dsfont}
\graphicspath{{Pictures/}}

\newtheorem*{remark}{Remark}
\newtheorem{theorem}{Theorem}
\newtheorem*{definition}{Definition}

\begin{document}

\pagestyle{empty}

\title{SWIPT using Hybrid ARQ over Time Varying Channels}

\author{Mehdi Salehi Heydar Abad, Ozgur Ercetin, Tamer ElBatt, Mohammed Nafie %
\thanks{Mehdi~Salehi~Heydar~Abad and Ozgur~Ercetin are with the Faculty of Engineering and Natural Sciences, Sabanci University, Istanbul, Turkey. Email: mehdis@sabanciuniv.edu, oercetin@sabanciuniv.edu}
\thanks{Tamer~ElBatt is with the Computer Science and Engineering Dept., The American 
University in Cairo, AUC Avenue, New Cairo 11835, Egypt and also Electronics and Communications Engineering Dept., 
Faculty of Engineering, Cairo University, Giza 12613, Egypt. Email:telbatt@ieee.org,}
\thanks{Mohammed~Nafie is with the Department of Electronics and Communications, Cairo University, and also Wireless Intelligent Networks Center, Nile University. Email:mnafie@ieee.org,}
\thanks{This work was in part supported by EC H2020-MSCA-RISE-2015 programme under grant number 690893, and also it was partially supported by a grant from the Egyptian National Telecommunications Regulatory Authority }
}
\maketitle
\newtheorem{lemma}{Lemma}
\newtheorem{corollary}{Corollary}
\thispagestyle{empty}

\begin{abstract}
We consider a class of wireless powered devices employing Hybrid Automatic Repeat reQuest (HARQ) to ensure reliable end-to-end communications over a two-state time-varying channel. A receiver, with no power source, relies on the energy transferred by a Simultaneous Wireless Information and Power Transfer (SWIPT) enabled transmitter to \emph{receive} and \emph{decode} information.  Under the two-state channel model, information is received at two different rates while it is only possible to harvest energy in one of the states.  The receiver aims to decode its messages with minimum expected number of re-transmissions. Dynamic and continuous nature of the problem motivated us to use a novel Markovian framework to bypass the complexities plaguing the conventional approaches such as MDP. Using the theory of absorbing Markov chains, we show that there exists an optimal policy utilizing the incoming RF signal solely to harvest energy or to accumulate mutual information. Hence, we convert the original problem with continuous action and state space into an equivalent one with discrete state and action space. For independent and identically distributed channels, we prove the optimality of a simple-to-implement {\em harvest-first-store-later} type policy. However, for time-correlated channels, we demonstrate that statistical knowledge of the channel may significantly improve the performance over such policies.
\end{abstract}

\section{Introduction}
\subsection{Background and Motivation}
 In traditional networks, wireless nodes are powered by limited capacity batteries which should be regularly charged or replaced. Energy harvesting has been recognized as a promising solution to replenish batteries without using any physical connections for charging.   In simultaneous wireless information and power transfer (SWIPT), the incoming RF signal is used for both energy harvesting and decoding of information bits. The inherent challenge of energy harvesting (EH) is the stochastic nature of the EH process, which dictates the amount and availability of harvested energy that is beyond the control of system designers. However, SWIPT may provide the network administrators a leverage on replenishing the remote devices for proper network operations.
 
In the seminal paper \cite{varshney2008transporting},  the rates at which energy and reliable information can be transferred over a single point-to-point noisy link  were characterized. This result was later extended to frequency-selective channels with additive white Gaussian noise (AWGN) in \cite{grover2010shannon}. In \cite{zhang2013mimo}, the authors examined separated and co-located information and energy receiver architectures in a multiple-input multiple-output (MIMO) wireless broadcast system. In separated architecture, both receivers have separate antennas, whereas in co-located architecture a single antenna is shared by both. In general, EH devices have small footprints necessitating a co-located architecture.  This arises a resource allocation problem of sharing the RF signal among the two receivers. The incoming RF signal is fed to Information Decoding (ID) and Energy Harvesting (EH) circuitries by applying either time-switching (TS) or power splitting (PS) schemes. In TS, the RF signal is split over two different parts of the time slot, one for EH and the other for ID, whereas in PS the incoming RF signal is fed to both, proportional to a given factor. In this work, we consider the class of PS policies.  In particular, we consider two types of PS policies: \emph{splitting} and \emph{no-splitting}. A splitting policy divides the RF signal into two parts with strictly non-zero power and feeds them to ID and EH circuitries, whereas no-splitting policy feeds the RF signal completely to either  EH or ID.

In inherently error-prone wireless communications systems, re-transmissions triggered by decoding errors have a major impact on the energy consumption of wireless devices. Hybrid automatic repeat request (HARQ) schemes are frequently used in order to reduce the number of re-transmissions by employing various channel coding techniques \cite{WHARQ}. Nevertheless, this comes at the expense of extra processing time and  energy associated with the enhanced error-correction decoders. A receiver employing HARQ encounters two major energy consuming operations: (1) sampling or Analog-to-Digital Conversion (ADC), which includes all RF front-end processing, and (2) decoding. The energy consumption attributed to sampling, quantization and decoding plays a critical role in energy-constrained networks which makes their study a non-trivial problem. The authors in \cite{de2014performance} investigated the performance of HARQ over an RF-energy harvesting point-to-point link, where the power transfer occurs over the downlink and the information transfer over the uplink. The authors studied the use of a TS policy when two HARQ mechanisms are used for information transfer; Simple HARQ (SH) and HARQ with Chase Combining (CC) \cite{harq}. Also, the authors in \cite{maha} studied the performance of HARQ in RF energy harvesting receivers, where heuristic TS policies are proposed to reduce the number of re-transmissions.  

In this paper, we consider a point-to-point link where an energy-abundant transmitter employs HARQ to deliver a message reliably to an EH receiver. The receiver has no energy source, so it relies on harvesting energy from the information-bearing RF signal. The channel is time-varying where the amount of  energy harvested and information collected varies depending on the quality of the channel.  The receiver aims to split the incoming RF signal between EH and ID so that the expected number of re-transmissions is minimized. Unlike prior works, e.g., \cite{zorzi}, in our work, we do not assume the availability of the channel state information (CSI) at the receiver\footnote{Due to the time and energy cost, the acquisition of CSI in EH networks is challenging. Some interesting ideas along this line, such as limited CSI feedback,  have been discussed in \cite{Clerckx}.}. 
\subsection{Contributions}
Our main contributions in this paper are summarized as follows:
\begin{itemize}
\item We formulate the problem of minimizing the expected number of re-transmissions using a Markov decision process (MDP). 
\item Due to the excessive number of states and actions in the MDP formulation, we use the special features of the EH HARQ framework to recast the MDP as a problem of minimizing the expected time to absorption in an absorbing Markov chain, significantly reducing the complexity associated with the MDP, when the wireless channel exhibits independent and identically distributed (i.i.d.),  and time-correlated properties, respectively.
\item For i.i.d. channels, we prove that there is an optimal policy that does not split the incoming RF energy and uses it solely either for ID or EH. As a result, we convert the original problem whose states and actions take over continuous values into  discrete ones, enabling a tractable solution.
\item The numerical solution of the MDP identifies multiple distinct policies that achieve the minimum expected number of re-transmissions, implying that the optimal policy is not unique.  Hence, we later completely characterize a class of simple-to-implement optimal policies.  Among those, harvest-first-store-later is an optimal policy lending itself for simple implementation on low complexity devices.
\item For a time-correlated channel, we once again show that there is an optimal policy that does not split the incoming RF energy. We develop a low complexity algorithm to determine the EH/ID decision for each state of the receiver.  Note that unlike the i.i.d. case, a simple policy such as harvest-first-store-later is no longer optimal for correlated channels as demonstrated in our numerical analysis.
\item We provide extensive numerical simulations to verify the analytical results established in the paper.
\end{itemize}


\subsection{Related Work}
Early works on wireless energy transfer \cite{early2} considered a point-to-point single antenna communication system and studied its rate-energy trade-off.  Single antenna systems are extended to single-input-multiple-output (SIMO) in \cite{liangliu}, multiple-input-single-output (MISO) in \cite{miso} and multiple-input-multiple-output (MIMO) system in \cite{mimo}. 

Note that EH devices harvest energy only in minuscule amounts (orders of $\mu W$s), so the energy consumption of the receiver circuitry to perform simple sampling and decoding can no longer be neglected. The authors in \cite{doost1} addressed  the energy consumption of sampling and decoding operations over a point-to-point link where the receiver harvests energy at a constant rate. In \cite{retarq2}, a decision-theoretic approach is developed to optimally manage the transmit energy of an EH transmitter transmitting to an EH receiver, where both the transmitter and the receiver harvests energy independently from a Bernoulli energy source. The receiver uses selective sampling (SS) and informs the transmitter about the SS information and its delayed battery state by feedback. Based on this feedback, the transmitter adjusts its transmission policy to minimize the packet error probability. 

Meanwhile, in \cite{doost2}, the performance of different HARQ schemes for an EH receiver harvesting energy from a deterministic energy source with a constant energy rate was studied.
In \cite{resist}, the impact of the battery's internal resistance at the receiver was analyzed for an EH receiver with imperfect battery, with the aim of maximizing the amount
of information decoded by the EH receiver.
While ignoring the sampling energy cost at the receiver, \cite{mehlul} investigates the performance of TS policies to maximize the amount of information decoded at the receiver operating over a binary symmetric channel (BSC), by optimizing the fraction of time used for harvesting energy and for extracting information.
For an EH transmitter and an EH receiver pair both harvesting ambient environmental energy with possible spatial correlation, \cite{outage_zhou} addresses the problem of outage minimization over a fading wireless channel with ACK-based re-transmission scheme by optimizing the power allocation  at the transmitter. 
In \cite{R3R1}, for a pair of EH transmitter-receiver employing ARQ and HARQ with binary EH process, packet drop probability over fading channels is minimized by optimally allocating power over different rounds of re-transmissions. In \cite{adaptarq}, an adaptive feedback mechanism for an EH receiver is proposed by taking into account the energy cost of sampling and decoding is proposed. The receiver is allowed to transmit a delayed feedback with the aim of efficiently utilizing the harvested energy in order to minimize the packet drop probability in the long run. In \cite{R3R4}, the outage probability for an EH receiver powered by RF transmissions is minimized by implementing HARQ. In particular, the transmitter optimally allocates two different power levels in charging and information transmission periods so that the probability of the event that information is not correctly received by the receiver due to either unsuccessful
message decoding or lack of minimum energy at the receiver is minimized. Although \cite{R3R4} is the most similar study to our work, it assumes that the channel stays constant during re-transmissions and it is known by the receiver. Differently, we assume that the wireless channel, with and without memory, varies over different instances of re-transmissions which calls for an online framework rather than an offline framework as in \cite{R3R4}. The problem of throughput optimization for an EH receiver operating in a multi-access network was studied in \cite{sufficient} where the receiver takes samples from the incoming RF signal to calculate the probability of a collision event and based on that decides to either utilize the incoming RF energy to replenish its battery or to extract information bits.



In \cite{mehdigilbert}, an EH transmitter intelligently adapts its channel sensing strategy with respect to a belief parameter it has about the channel condition to maximize its long term discounted throughput over a time correlated channel. In \cite{7492928}, maximization of long term weighted sum throughput, in an uplink scenario, for two RF EH transmitters is studied. The AP has the complete knowledge of the state of the network, i.e., battery levels, uplink and downlink CSI, and it calculates the optimal EH period, and the uplink durations of each transmitter at the beginning of each time slot. The finite horizon uplink throughput maximization for an EH transmitter with imperfect CSI and random EH process is studied in  \cite{7865904}, and the optimal power allocation problem at each time slot is formulated using dynamic programming (DP). \cite{R1} studies the rate-energy (R-E) region of separated and co-located SWIPT architectures where R-E region characterizes all the achievable rate and harvested energy pairs under a given transmit power constraint. A strategy  achieving the optimal R-E region is developed for the case of separated architecture. For the case of co-located architecture, two policies namely power splitting and time switching is investigated in terms of their achievable R-E region. In \cite{R2}, for a network with a transmitter, a relay and a destination node, two relaying protocols namely power splitting based relaying (PSR) and time switching based relaying (TSR)protocols are proposed. Analytical expressions for outage probability of delay limited transmission mode and ergodic capacity of delay tolerant transmission mode are derived. In contrast to \cite{R1,R2}, we show that there exists an optimal policy that does not split the incoming RF energy when HARQ mechanism is employed.

Differently from the available literature, we study the reliability of  transmission by an HARQ mechanism in a SWIPT scenario, over time varying channels with unknown CSI and by considering an accurate model of energy consumption of the EH receiver. We develop a novel Markovian framework for the analysis which facilitates characterizing the optimal decision at any given time. A major contribution of this work is that we prove that there exists an optimal no-splitting policy that minimizes the number of re-transmissions. This finding enables a tractable optimal solution by reducing a two dimensional uncountable state MC into a countable state MC. In particular, for i.i.d. channels, we show that policies such as harvest-first-store-later are optimal enabling simple-to-implement optimal policies suitable for low power EH devices. However, for the case of correlated channels, we show that an intelligent algorithm that utilizes the correlation information of the channel states, can significantly outperform those simple-to-implement policies.

\section{System Model and Preliminaries}\label{SystemModel}
\subsection{Channel Model and Receiver Architecture}
Consider a point-to-point time varying wireless link between a transmitter-receiver pair. 
The wireless channel is modeled according to a two-state block fading model where the states are GOOD and BAD\footnote{Note that the two-state channel process is an approximation of a more general multi-state time varying channel, where each state of the channel supports a maximum transmission rate. Here, we employ two-state channel process due to its analytical tractability.}.  Let $G_t\in\left\{0,\ 1\right\}$ be the state of the channel at time slot $t$ where BAD and GOOD states are denoted by $0$ and $1$, respectively. The CSI is neither available at the transmitter nor at the receiver due to the high computational and energy costs of transmitting and receiving a pilot signal necessary for measuring the CSI. We consider a  communication scheme where the transmitter is connected to a power source with an unlimited energy supply. The receiver is equipped with a separate rectifier circuit for EH and a transceiver for ID, both connected to the same antenna.



Time is slotted and each slot has a length of $N$ channel uses. We assume that $N$ is sufficiently large so that we can apply information theoretic arguments. The instantaneous achievable rate of the receiver is the maximum achievable mutual information between the output symbols of the transmitter and input symbols at the receiver. Let the achievable rate of the receiver be $R(t)$ at time $t$. As $N\to \infty$, $R(t)$ approaches the Shannon rate, and it can be computed as:
\small
\begin{align}
R(t) = \log_2(1 + P g(t)),
\end{align}
\normalsize
where $g(t)\in \left\{g_0, g_1\right\}$ is the channel power gain at time $t$ and $P$ is the noise-normalized transmit power of the transmitter. We assume that the transmitter power is fixed and known to the receiver. Let  $R_1$ and $R_0$ be the achievable rates corresponding to the channel states GOOD and BAD, respectively:
\small
\begin{align}
R_1 = \log_2(1 + P g_1),\label{R1R2a}\\
R_0 = \log_2(1 + P g_0).\label{R1R2b}
\end{align}
\normalsize
The instantaneous channel states are not known a priori so we employ an HARQ scheme with incremental redundancy (IR) for providing reliability \cite{wicker1995error}. In the following, we give a brief overview of HARQ-IR.




 \subsection{Brief Overview of HARQ}
 \label{sec:harq}
  HARQ is a well known method to provide reliable point to point communications \cite{wicker1995error}.  There are several types of HARQ implementations, e.g., simple HARQ, HARQ with Chase Combining (CC), repetition time diversity and incremental redundancy (IR). Note that in EH devices, CSI acquisition is cost prohibitive due to the energy and temporal cost of probing the channel. Hence, in this work, the transmitter is blind to the instantaneous channel conditions and it cannot adapt the code rates according to a particular channel gain. Thus, in our system, we consider HARQ-IR due to its superior throughput performance \cite{930931} compared to other alternatives as well as its robustness against the absence of CSI \cite{INRref}. Let us denote a message of the transmitter by $W\in \left\{1,2,\ldots,2^{NC}\right\}$, where $C$ denotes the rate of the information. Every incoming transport layer message into the transmitter is encoded by using a mother code of length $MN$ channel uses. The encoded message, $\mathbf{x}$, is divided into $M$ blocks, each of length $N$ channel uses, with a variable redundancy and it is represented by $\mathbf{x}=[x^1,\ldots,x^M]$. Let us assume that $x^1$ is transmitted at $t_1$. If $x^1$ is successfully decoded, then the receiver sends a 1-bit, error-free, zero-delay, Acknowledgement (ACK) message, otherwise, the transmitter times out after waiting a certain time period. In case of no ACK received, the transmitter transmits $x^2$ at time slot $t_2$ and the receiver combines the previous block $x^1$ with $x^2$. This procedure is repeated until the receiver accumulates $C$ bits of mutual information or maximum blocks of information, $M$, is sent. We assume that, $M$ is chosen sufficiently large so that the probability of decoding failure, due to exceeding the maximum number of re-transmissions, is approximately equal to zero. With HARQ-IR scheme, after $r$ re-transmissions, the amount of accumulated mutual information at the receiver is $\sum^{r}_{k=1} R(t_k)$. The receiver, given that it has sufficient energy, can perform a successful decoding attempt after $r$ re-transmissions, if the amount of accumulated mutual information exceeds the information rate of the transmitted message, i.e., $\sum^{r}_{k=1} R(t_k)\geq C$. We assume that each message is encoded at rate $R_1$ i.e., $C=R_1$ so that a transmission  in a GOOD channel state carries  all the information needed for decoding\footnote{Note that this assumption is practically reasonable, since a time slot is typically defined as the duration of time necessary for transmission of a single information packet.}.
  
\subsection{Energy Harvesting and Consumption Model}
In the following, we assume that the receiver has a sufficiently large battery and memory, so there is no energy or information overflow. The receiver utilizes a PS policy, where $\rho(t)\in [0,1]$ denotes the power splitting parameter at the beginning of time slot $t$. Note that $\rho(t)=0$ indicates that the received signal is used solely for mutual information accumulation, and $\rho(t)=1$ indicates that the received signal is used solely for harvesting energy. Any value of $\rho(t)$ between 0 and 1 refers to the case where the received signal is used for both harvesting energy and mutual information accumulation. 

We incorporate a simplified energy harvesting model, which facilitates the formulation of a tractable optimization problem.  In this model, the receiver harvests a maximum of $e\geq 1$  energy units in the GOOD channel state and zero units during the BAD channel state\footnote{The maximum energy is harvested if the received signal is completely directed to the energy harvester, i.e., $\rho(t)=1$.}. Typically, an EH device has two stages  in its energy harvesting circuitry \cite{Talla}: a rectifier stage that converts the incoming alternating current (AC) radio signals into direct current (DC); and a DC-DC converter that boosts the converted DC signal to a higher DC voltage value to produce the voltage required to charge the battery. The main limitation in an energy harvester is that every DC-DC converter has a minimum input voltage threshold below which it cannot operate. Hence, when the channel is in a BAD state, the input voltage is below the threshold of the DC-DC converter and no energy is harvested.  Even though the receiver cannot harvest any RF energy in a BAD channel state, it can still accumulate mutual information since ID circuit operates at a lower power sensitivity,  e.g., $-10$ dBm for EH and $-60$ dBm for ID circuits \cite{EHsurvey}.


The energy consumption of HARQ was recently investigated in \cite{rosas2016optimizing}, and it was identified that the energy is consumed at the start up of the receiver, during decoding, for operating passband receiver elements (low-noise amplifiers, mixers, filters,
frequency synthesizers, etc.), and for providing feedback to the transmitter. In order to develop a tractable optimization frame work, we consider the model in \cite{rosas2016optimizing}, and combine the individual costs of energy into two parameters only:  the receiver consumes $E_d\geq 1$ energy units for a decoding attempt and 1-energy unit for each mutual information accumulation event per time slot\footnote{One energy unit is normalized to the energy cost of operating the RF transceiver circuit during one time slot.}, i.e., operating the passband receiver elements.



\section{The Minimum Expected Number of Re-transmissions For I.I.D. Channels}
\label{MCF}
In this section, we calculate the minimum expected number of re-transmissions needed for successful decoding for time varying channels. We first consider an i.i.d. channel, and in Section VI, we will investigate the system under a time correlated channel model. Note that the receiver requires at least $E_d$ units of energy and $R_1$ bits of information before it can successfully decode the transmitted packet. Let the system states be  $(b,\ m)$, where $b$ is the total residual battery level and $m$ is the total accumulated mutual information normalized by $R_0$. For clarity of presentation, in the rest of the paper we assume that $R_0=1$.
Our objective is to optimally determine a scheduling policy $\rho(t)$ so that the transmission is successfully decoded with a minimum delay at the receiver. We formally define $\rho(t)$ next.


\begin{definition}
A scheduling policy  $\boldsymbol{\pi}= (\rho(1), \rho(2), \ldots,)$ is
a sequence of decision rules as such the $k$th element of $\boldsymbol{\pi}$
determines the power splitting ratio at $k$th time slot based on the observed system state $(b,\ m)$ at the beginning of this
time-slot for $t\in \{1,2,\ldots\}$. Similarly, a tail scheduling
policy $\boldsymbol{\pi}_t = (\rho(t), \rho(t+1), \ldots)$ is a sequence of decision rules
that determines the power splitting ratios for the time slots from $t$ to $\infty$.
\end{definition}

Let the probability that the channel is in GOOD
state be $\lambda$, i.e., $\mathds{Pr}\left[G_t=1\right]=\lambda$. The problem can be mathematically modeled as a two-state Markov chain (MC). Also, let the states of the MC be  $(b,\ m)$. It should be noted that the receiver is blind to the CSI before choosing the power splitting ratio. However, after it decides to sample the incoming RF signal for mutual information accumulation, the amount of the information in the sampled portion of the RF signal is revealed to the receiver. Because the scheduling policy is blind to the CSI, its decision only depends on $(b,\ m)$.

\subsection{Markov Decision Process (MDP) Formulation}
\label{sec:MDP}

At any given time $t$, the next state of the system only depends on the current state, $(b,\ m)$, and the power split ratio $\rho(t)$. Hence, we can formulate the problem as an MDP. Let $f^{\boldsymbol{\pi}}(t)\in\{-1,0\}$ be an indicator function taking a value of $0$ if the message can be decoded at the end of slot $t$ under policy $\boldsymbol{\pi}$, and a value of $-1$ otherwise. Then, the optimization problem we aim to solve is given as,
\small
 \begin{align}
 \max_{\boldsymbol{\pi}} \sum_{t=0}^{\infty} f^{\boldsymbol{\pi}}(t).
 \end{align}
\normalsize
 Let $V^{\boldsymbol{\pi}}(b,0)$ be the expected discounted reward with initial state $S_0 =(b,0)$ under policy $\boldsymbol{\pi}$ with discount factor $\beta \in [0, 1)$. The expected discounted reward has the following expression
\small
 \begin{align}
 V^{\pi}(b,\ 0) = \mathds{E}^{\boldsymbol{\pi}}\left[\sum^{\infty}_{t=0}\beta^{t}U(S_{t},\rho(t))|S_{0}=(b,\ 0)\right],\label{Vdef}
 \end{align}
\normalsize
 where $\mathds{E}^{\boldsymbol{\pi}}$ is the expectation with respect to the policy $\boldsymbol{\pi}$, $t$ is the time index, $\rho(t) \in [0,1]$ is the action chosen at time $t$, and $U(S_{t},\rho(t))$ is the instantaneous reward acquired when the current state is $S_t$.
In the rest of the paper, we use $\rho(t)$ and $\rho(b,m)$ interchangeably by assuming that at time slot $t$, the system is at state $(b,\ m)$. The battery is recharged with incoming RF signal depending on the value of the power split ratio $\rho(t)$.  Meanwhile, one unit of energy is consumed in order to accumulate non-zero bits of mutual information.  Hence, the evolution of the battery state is characterized as follows:
\small
\begin{align}
B(t)=&\left\{
\begin{array}{ll}
B(t-1)+\rho(t) e-\mathds{1}_{\rho(t)\neq 1},& \text{if}\  G_t=1\\
B(t-1)-\mathds{1}_{\rho(t)\neq 1},& \text{if}\  G_t=0
\end{array}
\right.,\label{B}
\end{align}
\normalsize
where $\mathds{1}_{\rho(t)\neq 1}=0$, if $\rho(t)=1$, and $\mathds{1}_{\rho(t)\neq 1}=1$, otherwise.

According to (\ref{R1R2a}) and (\ref{R1R2b}), the transmit power is equal to $P = \frac{2^{R_1}-1}{g_1} = \frac{2^{R_0}-1}{g_0}$. At the power splitter, $1-\rho(t)$ portion of the received power is directed into the ID, so the achievable mutual information accumulation is:
\small
\begin{align}
R(t) = \log_2(1 + g(t)P(1-\rho(t))).\label{splittedR}
\end{align} 
\normalsize
Note that the maximum value of the mutual information is attained by setting $\rho=0$. Inserting the value of $P$ in (\ref{splittedR}) for GOOD and BAD channel states gives the mutual information accumulation in these states respectively for a given power splitting ratio $\rho$ as
\small
\begin{align}
R^H(\rho)= \log_2(\rho+(1-\rho) 2^{R_1}),\\
R^L(\rho)= \log_2(\rho+(1-\rho) 2^{R_0}).
\end{align}
\normalsize
 Thus, the accumulated mutual information, $I(t)$, evolves as:
\small
\begin{align}
I(t)=&\left\{
\begin{array}{ll}
\min(I(t-1)+R^H(\rho(t)),R_1),& \text{if}\  G_t=1\\
\min(I(t-1)+R^L(\rho(t)),R_1),& \text{if}\  G_t=0
\end{array}
\right..\label{I}
\end{align}
\normalsize
Note that (\ref{I}) follows from the operation of HARQ-IR which is described in Section \ref{sec:harq} where the received messages over different time slots are combined in such a way that the mutual information of the combined messages is the summation of the individual mutual information of the messages. The instantaneous reward is zero if the message can be correctly decoded, and it is minus one otherwise. Recall that the decoding operation is successful if and only if the accumulated mutual information is above a certain threshold, and the battery level is sufficient to decode the message. Hence, the instantaneous reward is given as follows:
\small
 \begin{align}
U(S_{t},\rho(t))=&\left\{
 \begin{array}{ll}
 0,& \text{if}\  B_t\geq E_d,\ \text{and}\ I(t)\geq R_1,\\
 -1,& \text{if}\  \text{otherwise}.
 \end{array}
 \right..
 \end{align}
\normalsize
 Define the value function $V(b,m)$ as
\small
 \begin{align}
 V(b,\ m) &= \max_{\pi}V^{\pi}(b,\ m),\ \forall b\in[0,\infty),\ \forall m\in\left[0,\ R_1\right]\label{Vmax}.
 \end{align}
\normalsize
 The value function $V(b, m)$ satisfies the Bellman equation
\small
 \begin{align}
 V(b, m) = \max_{0\leq \rho \leq 1}V_\rho(b,m), \label{belman}
 \end{align}
\normalsize
 where $V_\rho(b,m)$ is the expected reward achieved by taking action $\rho$ when
 the state is $(b, m)$ and is given by
\small
 \begin{align}
 V_{\rho}(b,m)=U((b,\ m),\rho)+\beta\mathds{E}\left[V(\acute{b},\ \acute{m})|S=(b,\ m)\right], \label{actionvalue}
 \end{align}
\normalsize
 where $(\acute{b}, \acute{m})$ is the next visited state and the expectation is over the distribution of the next state. The use of expected discounted reward allows us to obtain a tractable solution, and one can gain insights into the optimal policy when $\beta$ is close to $1$. Value iteration algorithm (VIA) is a standard tool to solve Bellman equations such as the one in (\ref{belman}).  However, this problem suffers from the curse of dimensionality \cite{sutton}. Note that from (6) and (10), the problem is a two dimensional uncountable state  MDP with continuous actions at every state. Also, letting $\beta \rightarrow 1$, to approximate the average reward, slows down the algorithm to the point of infeasibility [30]. Hence, in the following, we take advantage of the special structure of our problem to derive an important characteristic of the optimal policy. The flow of the paper is depicted in Figure \ref{fig:diagram}.

 \begin{figure}[ht]
  \centering
    \includegraphics[scale=0.2]{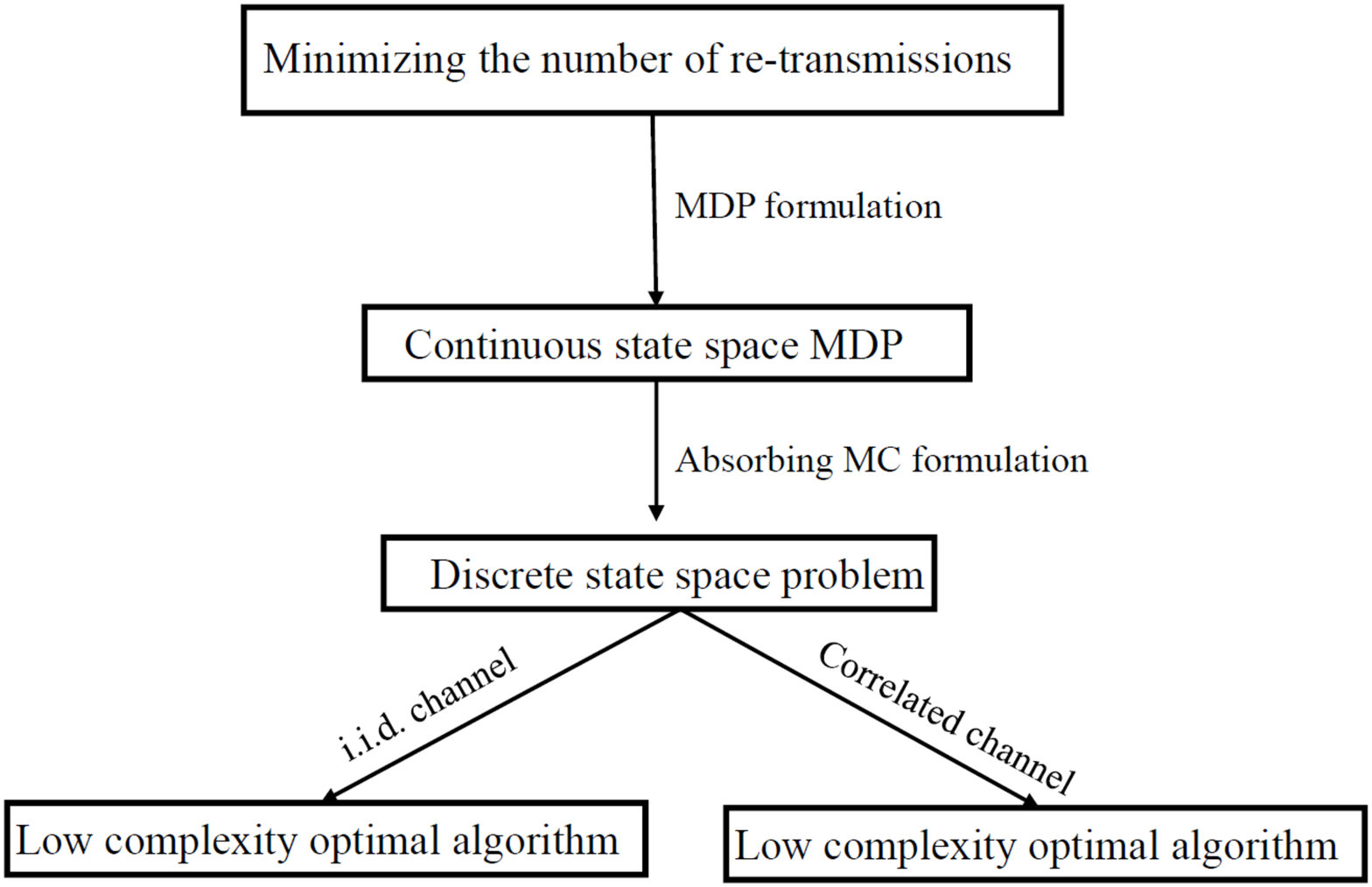}
		  \caption{A brief overview of the paper.}
			\label{fig:diagram}
\end{figure}

\subsection{Absorbing Markov Chain Formulation}
Note that the MC describing the operation of our system is an {\em absorbing} MC, where all states except those $(b, m)$ where $b\geq E_d$, and $m\geq R_1$ are transient states.  The absorbing states are those where the receiver has both sufficient energy and information accumulated to correctly decode. In an absorbing chain, starting from a transient state, the chain makes a finite number of visits to some transient states before its eventual absorption into one of the absorbing states. Hence, the mean time to absorption of the chain, starting from transient state $i$  initially, is the sum of the expected numbers of visits made to transient states. In an absorbing MC, the expected number of steps taken before being absorbed in an absorbing state characterizes the \emph{mean time to absorption}. Hence, the mean time to absorption starting from a given transient state $(b,\ m)$ provides the number of re-transmissions until successful decoding when the battery has $b$ units of energy and the memory contains $m$ bits of information. 


 After establishing the $\rho$ dependent state evolution of $B(t)$ and $I(t)$, we can formally introduce the state transition probabilities of the Markov chain as follows:

 \begin{align}
 \rho=1\Rightarrow&\left\{
 \begin{array}{ll}
 \mathds{Pr}\left((B,I),(B+l,I)\right) = \lambda\\
 \mathds{Pr}\left((B,I),(B,I)\right) = 1-\lambda
 \end{array}
 \right.,
 \end{align}
 \begin{align}
 \rho=0\Rightarrow&\left\{
 \begin{array}{ll}
 \mathds{Pr}\left((B,I),(B-1,R_1)\right) = \lambda\\
 \mathds{Pr}\left((B,I),(B-1,I+1)\right) = 1-\lambda
 \end{array}
 \right.,
 \end{align}

 \begin{align}
 0<\rho<1\Rightarrow&\left\{
 \begin{array}{ll}
 \mathds{Pr}\left((B,I),(B-1+\rho l,I+R^H(\rho))\right) = \lambda\\
 \mathds{Pr}\left((B,I),(B-1,I+R^L(\rho))\right) = 1-\lambda
 \end{array}
 \right.,
 \end{align}
where $\mathds{Pr}(\mathbf{x},\mathbf{y})$ is the transition probability from state $\mathbf{x}$ into state $\mathbf{y}$, $B\in [0,\infty)$ and $I\in [0,\ R_1]$. The state transition probabilities of the Markov chain associated with $\rho$ is depicted in Figure \ref{fig:DTMC}.

 \begin{figure}
         \centering
         \begin{subfigure}[t]{0.15\textwidth}
                 \includegraphics[width=\textwidth]{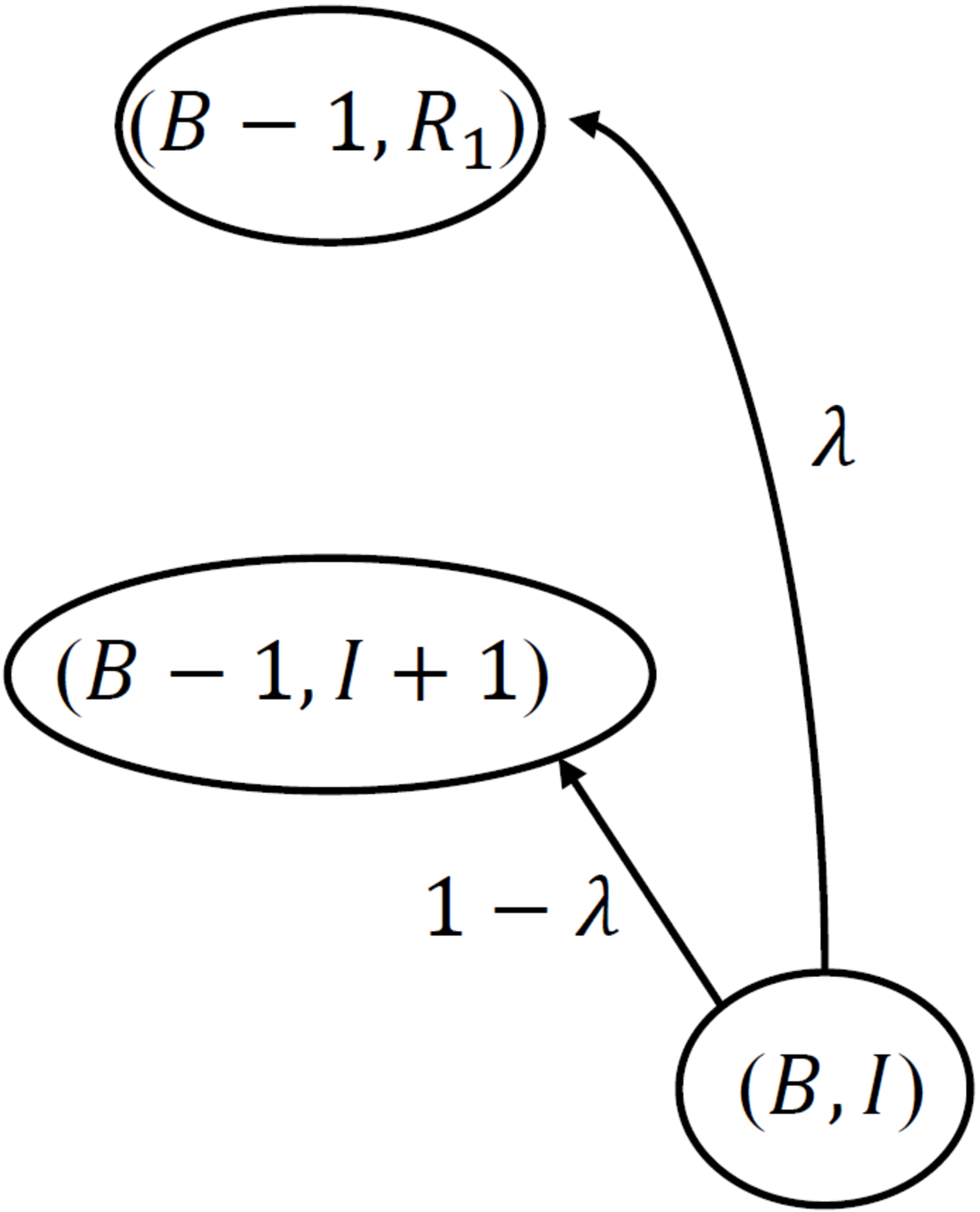}
                 \caption{$\rho = 0$.}
                 \label{fig:DTMC1}
         \end{subfigure}
         \hfill
         \begin{subfigure}[t]{0.15\textwidth}
                 \includegraphics[width=\textwidth]{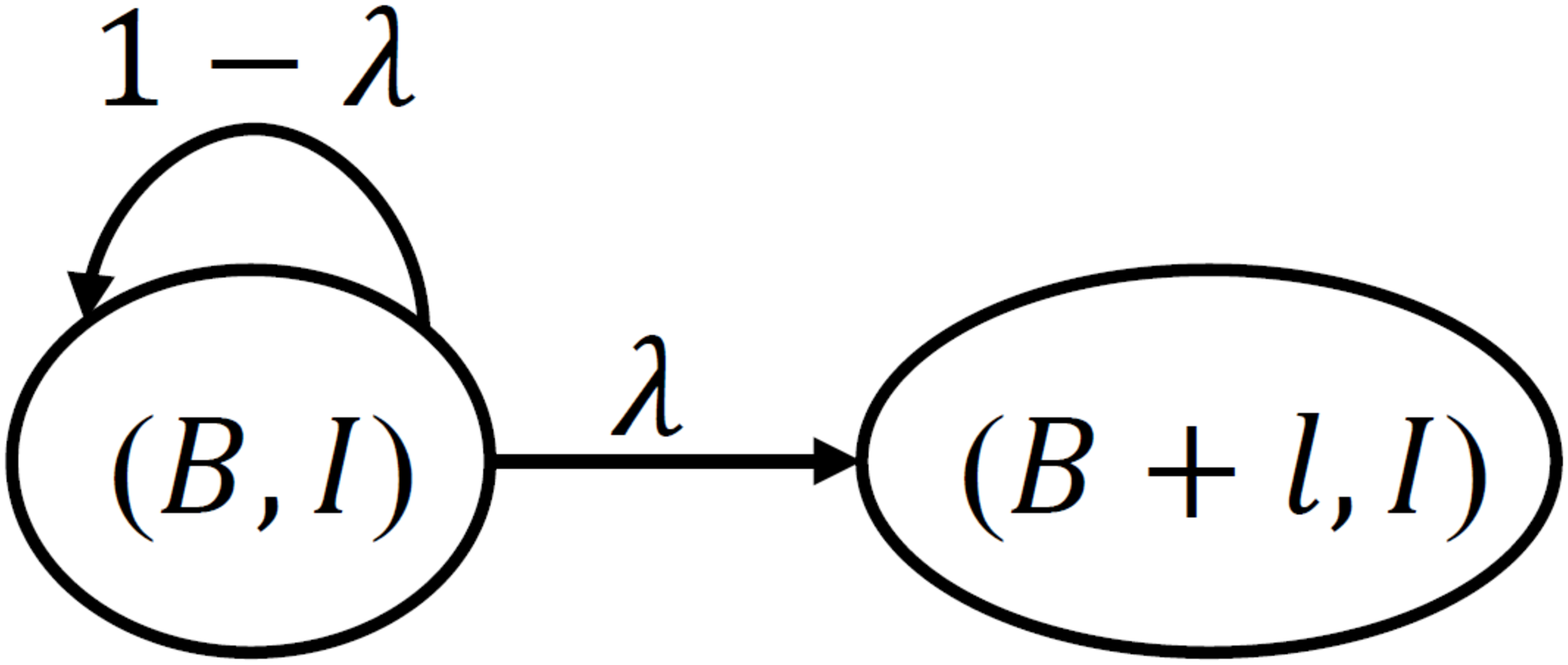}
                 \caption{$\rho = 1$. }
                 \label{fig:DTMC2}
         \end{subfigure}
                 \begin{subfigure}[t]{0.15\textwidth}
                 \includegraphics[width=\textwidth]{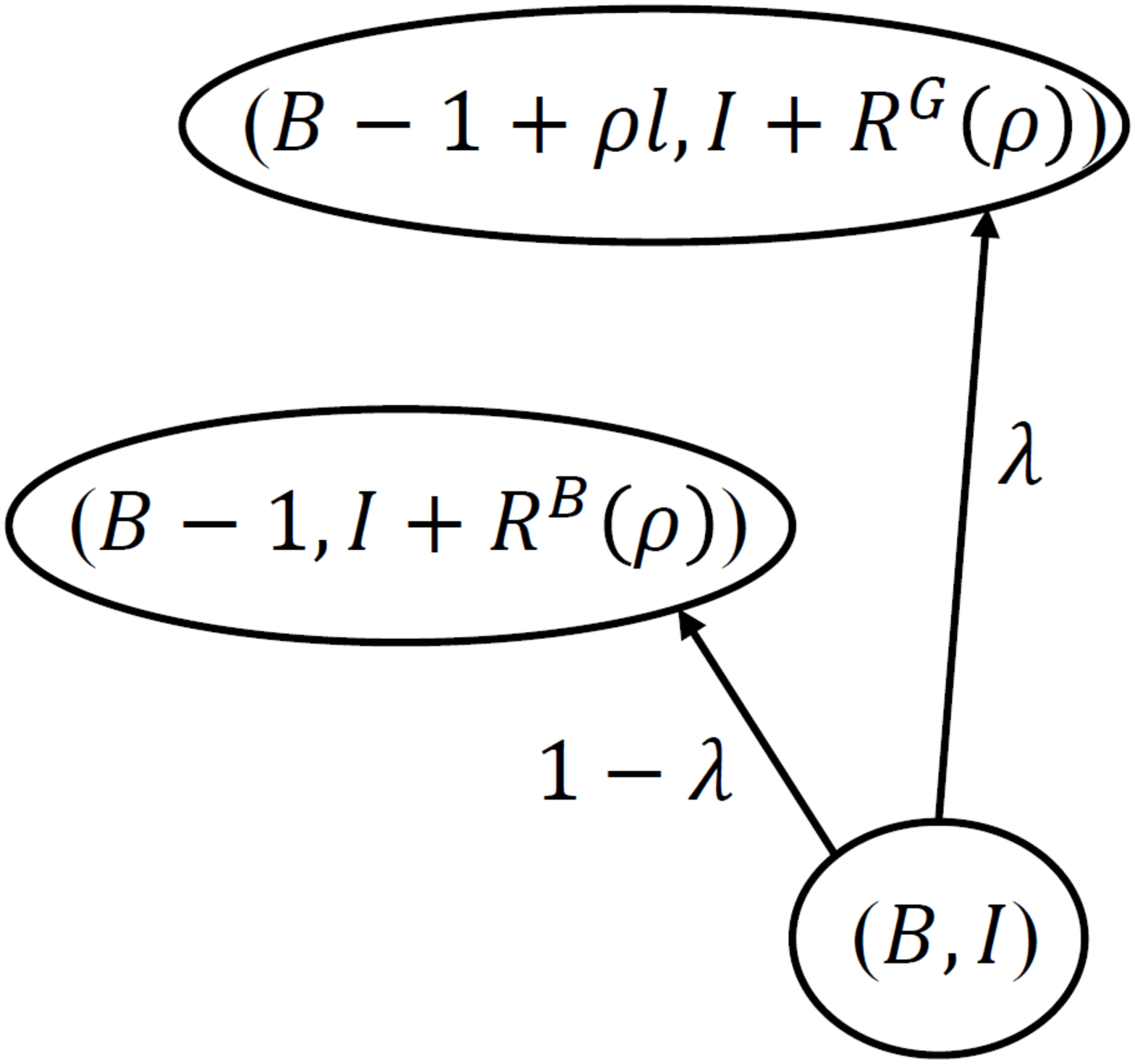}
                 \caption{$0<\rho<1$.}
                  \label{fig:DTMC3}
         \end{subfigure}
 	                \caption{State transition probabilities of the Markov chain associated with                                   $\rho$.}
                  \label{fig:DTMC}
              
 \end{figure}
In the following, we perform \emph{first-step analysis}, by conditioning on the first step
the chain makes after moving away from a given initial state to obtain the mean time to absorption. Let $k_{b,m}$ be the expected number of transitions needed to hit
an absorbing state when the MC starts from state $(b,\ m)$. The analysis is performed by assuming that the MC is in steady-state.

Let us first consider two trivial cases; when the battery has less than one unit of energy, i.e., $b<1$, in which case the receiver has no option but harvest the incoming RF signal, and when the amount of accumulated mutual information is $R_1$, in which case there is no point in further accumulating mutual information since the receiver has sufficient mutual information to decode the incoming packet.  For these cases, the mean time to absorption starting from an initial state $(b,\ m)$ is
\begin{equation} \label{kBless1}
\begin{split}
k_{b,m} &= 1+ \lambda k_{b+e,m}+(1-\lambda) k_{b,m} \\
 &=\frac{1}{\lambda} + k_{b+e,m},\hspace{1cm} \text{if}\ b<1\ \text{or}\ m = R_1.
\end{split}
\end{equation}
Note that in (\ref{kBless1}), one slot is needed to harvest energy, and depending on the channel state in that slot,  the battery state either transitions to $b+e$ or remains the same. The following lemma plays an important role in establishing the structure of the optimal policy.

\begin{lemma}
\label{lemma1}
For any $E_d-i\cdot e \leq b < E_d-(i-1)\cdot e$ such that $i=1,\ldots,E_d$, given that $m=R_1$, the mean time to absorption is given by, $k_{b,R_1} = \frac{i}{\lambda}$.
\end{lemma}
\begin{proof}
The proof is given in Appendix \ref{prooflemma1}.
\end{proof}

We will use Lemma \ref{lemma1}  to show that the optimal policy minimizing the mean time to absorption \emph{does not} need to split the incoming RF signal. In order to show this, let us define two tail policies $\boldsymbol{\pi}^i_t=(a_i,\boldsymbol{\pi}_{t+1})$, $i=split,no-split$ taking different actions $a_i$, in the current slot, but following the same set of actions, $\boldsymbol{\pi}_{t+1}$ afterwards\footnote{Note that $(a_i,\boldsymbol{\pi}_{t+1})$ defines a tail policy obtained by concatenating action $a_i$ in the current slot with tail policy $\boldsymbol{\pi}_{t+1}$.}. Let policy $\boldsymbol{\pi}^{split}_t=(\rho, \boldsymbol{\pi}_{t+1})$ be a tail policy that always splits the incoming RF energy, i.e., $0<\rho<1$, except when $B(t)<1$ or $I(t)=R_1$, when it only harvests energy.
Assume that the state of the system is $(b,\ m)$ at time slot $t$. Then, the mean time to absorption for tail policy $\boldsymbol{\pi}^{split}_t$ is:
\small
\begin{align}
k^{\boldsymbol{\pi}^{split}}_{b,m} = 1 + \lambda k_{b-1+\rho e,m+R^H(\rho)} + (1-\lambda) k_{b-1,m+R^L(\rho)},\label{k_pi}
\end{align}
\normalsize
where $k_{x,y}$ is the mean time to absorption of policy $\boldsymbol{\pi}_{t+1}$ beginning at state $(x,y)$. Note that with probability $\lambda$ the channel is in  GOOD state, and thus, $\rho\cdot e$ units of energy is harvested\footnote{We assume that the energy harvesting circuit is generating energy linearly proportional to the energy of the incoming RF signal.}. However, one unit of energy is spent by operating the transceiver to accumulate $R^H(\rho)$ bits of mutual information. Meanwhile, with probability $1-\lambda$ the channel is in  BAD state, and no energy is harvested, but the transceiver still consumes one unit of energy to accumulate $R^L(\rho)$ bits of mutual information.

Under tail policy $\boldsymbol{\pi}^{no-split}_t$ the RF signal is never split at time slot $t$, but rather, it is completely used for  mutual information accumulation except when $B(t)<1$ or $I(t) = R_1$ when it harvests energy only. In a similar way as before, we may calculate $k^{\boldsymbol{\pi}^{no-split}}_{b,m}$ as follows:
\small
\begin{align}
k^{\boldsymbol{\pi}^{no-split}}_{b,m} = 1 + \lambda k_{b-1,R_1} + (1-\lambda)k_{b-1,m+ R_0}.\label{kalpha}
\end{align}
\normalsize
\begin{theorem}
\label{TH_dis}
Policy $\boldsymbol{\pi}^{no-split}_t$ in \eqref{kalpha} achieves an expected number of re-transmission that is never worse than that of policy $\boldsymbol{\pi}^{split}_t$ in (19), i.e., $k^{\boldsymbol{\pi}^{no-split}}_{b,m}\leq k^{\boldsymbol{\pi}^{split}}_{b,m}$ for every $b=0,1,\ldots$ and $m=0,1,\ldots,R_1$.
\end{theorem}
\begin{proof}
The proof is given in Appendix \ref{proofTH_dis}.
\end{proof}
Theorem \ref{TH_dis} proves that a no-splitting policy can achieve the minimum number of re-transmissions. Hence, in the latter part of the paper, we focus on characterizing the optimal no-splitting policy by determining the scheduling decision between EH or ID for each state of the MC.  Therefore, the state space of the discrete MC associated with the optimal no-splitting policy is  $b = 0,1,\ldots,\infty$, and $m = 0,1,\ldots,R_1$\footnote{Note that in the original problem the states of the MC are $[0,\infty)\times[0,R_1]$.}. 
 \begin{remark}
Theorem \ref{TH_dis} plays an important role in simplifying the original problem by reducing the two dimensional uncountable state MDP with continuous action space into a two dimensional countable state MDP with binary decision space. This significantly reduces the complexity of numerical methods such as VIA. However, as we shall see in Section \ref{sec:iid}, the absorbing MC framework helps prove the optimality of a class of simple-to-implement algorithms that is more suitable for resource-deficient EH devices. 
 \end{remark}

Since the class of policies that we are interested in does not observe the channel, but make a decision based only on $(b,\ m)$, the time of the decision is irrelevant.  Hence, given $(b,\ m)$, time $t$ and $t+1$ are stochastically identical. Therefore, in the rest of the paper we will omit the time index and optimize the scheduling decisions for any given state $(b,\ m)$. Define $\boldsymbol{\pi}^*$ as the optimal policy minimizing the mean time to absorption beginning at any given state $(b,\ m)$. Let $k^{\boldsymbol{\pi}^*}_{b,m}$ be the minimum mean time to absorption obtained by policy $\boldsymbol{\pi}^*$\footnote{Note that the mean time to absorption calculated in Lemma \ref{lemma1} is the smallest possible value, i.e., $k^{\boldsymbol{\pi}^*}_{b,R_1}=k_{b,R_1}$ for $b=0,1,\ldots, E_d-1$.}. Define the tail policy $\boldsymbol{\pi}^{i}(b,m) = (i,\boldsymbol{\pi}^*(\acute{b},\ \acute{m}))$, $i=0,1$ such that it chooses $\rho=i$ at state $(b,\ m)$ but follows policy $\boldsymbol{\pi}^*$ after transitioning into the new state $(\acute{b},\ \acute{m})$. Let $k^{\boldsymbol{\pi}^i}_{b,m}$ be the mean time to absorption of policy $\boldsymbol{\pi}^{i}(b,m)$, $i=0,1$. We can characterize $k^{\boldsymbol{\pi}^0}_{b,m}$ and $k^{\boldsymbol{\pi}^1}_{b,m}$ as follows:
\small
\begin{align}
k^{\boldsymbol{\pi}^0}_{b,m} &= 1 + \lambda k^{\boldsymbol{\pi}^*}_{b-1,R_1} + (1-\lambda)k^{\boldsymbol{\pi}^*}_{b-1,m+1},\label{k0}\\
k^{\boldsymbol{\pi}^1}_{b,m} &= 1 + \lambda k^{\boldsymbol{\pi}^*}_{b+e,m} + (1-\lambda)k^{\boldsymbol{\pi}^1}_{b,m}\nonumber\\
&=\frac{1}{\lambda} + k^{\boldsymbol{\pi}^*}_{b+e,m}\label{k1}.
\end{align}
\normalsize
Note that by evaluating and then comparing the values of $k^{\boldsymbol{\pi}^0}_{b,m}$ and $k^{\boldsymbol{\pi}^1}_{b,m}$, at all possible states $(b,\ m)$ for $b = 0,1,\ldots,\infty$, and $m = 0,1,\ldots,R_1$, one can obtain the optimal policy $\boldsymbol{\pi}^*$ and its associated $k^{\boldsymbol{\pi}^*}_{b,m}$.

\begin{theorem}
\label{lemma2}
For states $(b,\ m)=(E_d+j,\ R_1-j)$ for $j = 1,2,\ldots,R_1$, the minimum mean time to absorption, $k^{\boldsymbol{\pi}^*}_{b,m}$ is given by
\small
\begin{align}
k^{\boldsymbol{\pi}^*}_{E_d+j,R_1-j} = k^{\boldsymbol{\pi}^0}_{E_d+j,R_1-j} = \sum^{j}_{i=1} (1-\lambda)^{i-1}.
\end{align}
\normalsize
Furthermore, $k^{\boldsymbol{\pi}^*}_{b,R_1-j}=k^{\boldsymbol{\pi}^0}_{E_d+j,R_1-j}$ for $b=E_d+j+1, E_d+j+2,\ldots$.
\end{theorem}
\begin{proof}
The proof is given in Appendix \ref{prooflemma2}.
\end{proof}
Theorem \ref{lemma2} states that if the receiver has $R_1-n$ bits of mutual information accumulated and more than $E_d+n$ units of energy in its battery, then it should use the incoming RF signal for mutual information accumulation only. For any given state $(b,\ m)$, we exploit Lemma \ref{lemma1} and Theorem \ref{lemma2} to develop Algorithm \ref{alg1} for calculating the minimum mean time to absorption, $k^{\boldsymbol{\pi}^*}_{b,m}$, and the optimal scheduling decision at every state.

The idea of Algorithm \ref{alg1} is to use Lemma \ref{lemma1} and Theorem \ref{lemma2} as boundary conditions and to recursively calculate the mean time to absorption $k^{\boldsymbol{\pi}^0}_{b,m}$ and $k^{\boldsymbol{\pi}^1}_{b,m}$ starting from $(b,\ m)=(E_d,R_1 -1)$. Note that $k^{\boldsymbol{\pi}^0}_{E_d,R_1 -1}$ and $k^{\boldsymbol{\pi}^1}_{E_d,R_1 -1}$ depend on the values of $k^{\boldsymbol{\pi}^*}_{E_d-1,R_1}$ and $k^{\boldsymbol{\pi}^*}_{E_d+1,R_1 -1}$, which are obtained in the initialization step, and the optimal scheduling decision at state $(E_d, R_1-1)$ is given by $\arg\min_{i\in{0,1}} k^{\boldsymbol{\pi}^i}_{b,m}$. The procedure in Algorithm \ref{alg1} continues by decrementing the value of $b$ by $1$ at each iteration, until $b=0$ at which time the value of $m$ is decremented by $1$, $b$ is initialized to $E_d+n$ and the procedure is repeated.  The aforementioned order of spanning the states of the MC ensures that at each iteration the mean time to absorption can be calculated from the values determined in the previous iterations. We have shown in Appendix \ref{proofalg1}, that Algorithm \ref{alg1} minimizes the expected number of re-transmissions starting from any state $(b,\ m)$.

\begin{algorithm}
\caption{Calculating the minimum mean time to absorption for an i.i.d. channel}\label{alg1}
\begin{algorithmic}[1]
\State Initialize $k^{\boldsymbol{\pi}^*}_{b,R_1}$ for $b = 0,\ldots,E_d-1$ using Lemma \ref{lemma1}.
\State Initialize $k^{\boldsymbol{\pi}^*}_{E_d+j,R_1-j}$ for $j = 1,\ldots,R_1$ using Theorem \ref{lemma2}.
\State $n\gets 0$
\For{$m = R_1-1:0$}
\For{$b = E_d+n:0$}
\State Calculate $k^{\boldsymbol{\pi}^0}_{b,m}$, $k^{\boldsymbol{\pi}^1}_{b,m}$ from (\ref{k0}) and (\ref{k1}), respectively.
\State $k^{\boldsymbol{\pi}^*}_{b,m}=\min\left(k^{\boldsymbol{\pi}^0}_{b,m}, k^{\boldsymbol{\pi}^1}_{b,m}\right)$.
\State $\rho^*(b,m) = \arg\min_i k^{\boldsymbol{\pi}^i}_{b,m}$ for $i = 0,1$
\EndFor
\State $n\gets n+1$
\EndFor
\end{algorithmic}
\end{algorithm}


\section{Optimal Class of Policies for i.i.d. Channels}
\label{sec:iid}

In the previous section, we have given a procedure to obtain the optimal scheduling decision of a no-splitting policy, once we established that there exists a no-splitting policy achieving the minimum number of re-transmissions. In this section, we formally determine the optimal class of scheduling policies minimizing the number of re-transmissions until successful decoding. In the following, we obtain our analytical results for $e=1$ and $R_0=1$. However, our analysis holds in general for different values of $e$ and $R_0$, as demonstrated by the numerical results presented in Section \ref{Results}. Note that once the battery has sufficient charge to decode the packet, i.e., $b=E_d+1,E_d+2,\ldots$, it is better to use the incoming RF signal only for information accumulation.
For the remaining states, i.e., $b=1,2,\ldots, E_d$, and $m=0,1,\ldots R_1-1$,  any scheduling decision, either $\rho=0$ or $\rho=1$, is optimal. These two facts are proven formally in Appendix \ref{proofTH_str1} and \ref{proofTH_str2} respectively. This result, in essence, proves that there is no unique optimal policy. Instead, there exists a family of optimal policies achieving the same minimum mean time to absorption. We summarize our findings so far in the following theorem by formally characterizing the family of optimal policies.

\begin{theorem}
\label{TH_str}
Optimal policy, $\boldsymbol{\pi}^*$, satisfies the following properties.
\begin{enumerate}
\item If $b=0$ or $m=R_1$, it chooses $\rho=1$.
\item If $b=E_d+1,E_d+2,\ldots$, it chooses $\rho=0$.
\item If $b=1,2,\ldots, E_d$ and $m=0,1,\ldots R_1-1$, chooses either $\rho=0$ or $\rho=1$.
\end{enumerate}
\end{theorem}
\begin{proof}
The proof is given in Appendix \ref{proofTH_str}.
\end{proof}

Simple examples of such optimal policies that belong to the optimal family of policies characterized in Theorem \ref{TH_str}, are:

\begin{itemize}
\item Battery First (BF): the receiver harvests energy until it acquires $E_d$ units of energy and then starts accumulating the mutual information.
\item Information First (IF): the receiver always accumulates mutual information 
unless $b=0$ or $m=R_1$.
\item Coin Toss (CT): the receiver harvests energy when $b=0$ or $m=R_1$, while it accumulates mutual information when $b=E_d+1,E_d+2,\ldots$. Otherwise, it tosses a fair coin to choose between harvesting energy or accumulating mutual information.
\end{itemize}



\section{Expected Number of Re-Transmissions for a Correlated Channel}
\label{sec:cor}
In many wireless systems, the wireless channel cannot be modeled as an i.i.d. channel. In this section, we investigate optimal
scheduling policies under a time-correlated channel model. Our analysis for a correlated channel follows a similar approach to our analysis for i.i.d. channels.  However, due to correlation between the subsequent channel states, the receiver can improve its decision by incorporating its knowledge of the current state.  
Let the transition probabilities of the channel states be $\mathds{Pr}\left[G_t=1|G_{t-1}=1\right]=\lambda_1$ and $\mathds{Pr}\left[G_t=1|G_{t-1}=0\right]=\lambda_0$. Note that due to time correlation, the previous state of the channel provides information about the current channel state to the receiver. Hence, although once again we model the system as a MC, this time the state space of MC is extended where the states are $(b,\ m,\ G)$  with $G$ being the previous state of the channel\footnote{Note that the receiver  becomes aware of the channel state after it decides to sample the incoming RF signal.}. The resulting MC is still an absorbing MC, and the mean time to absorption is equivalent to the minimum expected number of re-transmissions until successful decoding. Define $\boldsymbol{\pi}^*$ as the optimal policy minimizing the mean time to absorption at any given state $(b,\ m,\ G)$. Let $k^{\boldsymbol{\pi}^*}_{b,m,G}$ be the mean time to absorption obtained by policy $\boldsymbol{\pi}^*$ at state $(b,\ m,\ G)$. 

\begin{lemma}
\label{cor1}
For any $E_d-i\cdot e \leq b < E_d-(i-1)\cdot e$ such that $i=1,\ldots,E_d$, and given that $m=R_1$, the minimum mean time to absorption is given by 
\small
\begin{align}
&k^{\boldsymbol{\pi}^*}_{b,R_1,1} =i \frac{1+\lambda_0-\lambda_1}{\lambda_0}, &i=1,\ldots,E_d,\label{eq:alg1}\\
&k^{\boldsymbol{\pi}^*}_{b,R_1,0} =\frac{1}{\lambda_0} + (i-1) \frac{1+\lambda_0-\lambda_1}{\lambda_0},  &i=1,\ldots,E_d.\label{eq:alg2}
\end{align}
\normalsize
\end{lemma}
\begin{proof}
The proof is similar to that of Lemma \ref{lemma1}. The detailed proof is given in \cite{arxivmehdi}.
\end{proof}

Similar to Theorem \ref{TH_dis}, by exploiting  Lemma \ref{cor1}, we can prove that the optimal policy should either choose energy harvesting or information accumulation at any given state $(b,\ m,\ G)$. Therefore, MC associated with the optimal strategy has discrete states in which $b = 0,1,\ldots,\infty$, $m = 0,1,\ldots,R_1$ and $G = 0,1$. Define the tail policy $\boldsymbol{\pi}^{i}(b,m,G) = (i,\boldsymbol{\pi}^*(\acute{b},\ \acute{m},\ \acute{G}))$, $i=0,1$ that chooses $\rho=i$ at state $(b,\ m,\ G)$ but follows policy $\boldsymbol{\pi}^*$ after transitioning into the new state $(\acute{b},\ \acute{m},\ \acute{G})$. Let $k^{\boldsymbol{\pi}^i}_{b,m,G}$ be the mean time to absorption of policy $\boldsymbol{\pi}^{i}(b,m,G)$, $i=0,1$. We can calculate $k^{\boldsymbol{\pi}^0}_{b,m,G}$ and $k^{\boldsymbol{\pi}^1}_{b,m,G}$ as follows:
\small
\begin{align}
&k^{\boldsymbol{\pi}^0}_{b,m,0}= 1+ \lambda_0 k^{\boldsymbol{\pi}^*}_{b-1,R_1,1} + (1-\lambda_0)k^{\boldsymbol{\pi}^*}_{b-1,m+1,0}\label{inf0},\\
&k^{\boldsymbol{\pi}^0}_{b,m,1}= 1+ \lambda_1 k^{\boldsymbol{\pi}^*}_{b-1,R_1,1} + (1-\lambda_1)k^{\boldsymbol{\pi}^*}_{b-1,m+1,0}\label{inf1},\\
&k^{\boldsymbol{\pi}^1}_{b,m,0}= 1+ \lambda_0 k^{\boldsymbol{\pi}^*}_{b+1,m,1} + (1-\lambda_0)k^{\boldsymbol{\pi}^1}_{b,m,0}=\frac{1}{\lambda_0} + k^{\boldsymbol{\pi}^*}_{b+1,m,1},\label{eh0}\\
&k^{\boldsymbol{\pi}^1}_{b,m,1}= 1+ \lambda_1 k^{\boldsymbol{\pi}^*}_{b+1,m,1} + (1-\lambda_1)k^{\boldsymbol{\pi}^*}_{b,m,0}.\label{eh1}
\end{align}
\normalsize

Similar to the outline of the Theorem \ref{lemma2}, in the following, we consider states $(b,\ m,\ G)=(E_d+j,R_1-j,G)$ for $j=1,\ldots,R_1$ and derive the optimal strategy for those states. 

\begin{lemma}
\label{cor-lema}
The optimal strategy in states $(E_d+j,R_1-j,G)$ for $j=1,\ldots,R_1$ and $G = 0,1$ is to accumulate mutual information ($\rho^*(E_d+j,R_1-j,G) = 0$) and also $k^{\boldsymbol{\pi}^*}_{b,R_1-j,G} =k^{\boldsymbol{\pi}^0}_{E_d+j,R_1-j,G}$ for $b= E_d+j+1,E_d+j+2,\ldots$.

\end{lemma}
\begin{proof}
The proof is similar to that of Theorem \ref{lemma2}. The detailed proof is given in \cite{arxivmehdi}.
\end{proof}
Now that we know the optimal policy for states $(E_d+j,R_1-j,G)$, we can calculate the minimum mean time to absorption for those states as follows:
\small
\begin{align}
&k^{\boldsymbol{\pi}^*}_{E_d+j,R_1-j,0} =k^{\boldsymbol{\pi}^0}_{E_d+j,R_1-j,0}= \sum^{j}_{i=1}(1-\lambda_0)^{i-1},\hspace{0.5cm} j=1,\ldots,R_1,\label{eq:alg3}\\
&k^{\boldsymbol{\pi}^*}_{E_d+j,R_1-j,1} =k^{\boldsymbol{\pi}^0}_{E_d+j,R_1-j,1}=1+ (1-\lambda_1)\sum^{j-1}_{i=1}(1-\lambda_0)^{i-1},\nonumber\\ &\hspace{6cm} j=2,\ldots,R_1,\label{eq:alg4}\\
&k^{\boldsymbol{\pi}^*}_{E_d+j,R_1-j,1} = 1,\hspace{3.65cm} j=1.\label{eq:alg5}
\end{align}
\normalsize

Algorithm \ref{alg2} calculates the $k^{\boldsymbol{\pi}^*}_{b,m,G}$ and the corresponding $\rho^*$ for any $b$, $m$, and $G$. Proving the optimality of Algorithm \ref{alg2} is similar to the outline of the optimality proof of Algorithm \ref{alg1} and hence it is omitted here. Note that the knowledge of the previous channel state, $G$, enables the receiver to be able to fully utilize the information yielded by the correlation. However, it also results in four coupled equations, (\ref{inf0})-(\ref{eh1}), over numerous states which makes the analysis extremely hard. For this reason, we omit the full characterization of the structure of the optimal policy.   Nevertheless, note that Algorithm \ref{alg2} provides a recursive method to determine the optimal scheduling decisions for each state $(b,m, G)$.  In fact, we use these optimal decisions in the numerical experiments discussed in Section \ref{Results} to calculate the minimum number of re-transmissions.

\begin{algorithm}
\caption{Calculating the minimum mean time to absorption for correlated channel}\label{alg2}
\begin{algorithmic}[1]
\State Initialize $k^{\boldsymbol{\pi}^*}_{b,R_1,G}$ for $b = 0,\ldots,E_d-1$ using (\ref{eq:alg1}) and (\ref{eq:alg2}).
\State Initialize $k^{\boldsymbol{\pi}^*}_{E_d+j,R_1-j,G}$ for $j = 1,\ldots,R_1$ using (\ref{eq:alg3}), (\ref{eq:alg4}) and (\ref{eq:alg5}).
\State $n\gets 0$
\For{$m = R_1-1:0$}
\For{$b = E_d+n:0$}
\State Calculate $k^{\boldsymbol{\pi}^0}_{b,m,G}$ for $G=0,1$ using (\ref{inf0}) and (\ref{inf1}), respectively.
\State Calculate $k^{\boldsymbol{\pi}^1}_{b,m,G}$ for $G=0,1$ using (\ref{eh0}) and (\ref{eh1}), respectively.
\State $k^{\boldsymbol{\pi}^*}_{b,m,G}=\min\left(k^{\boldsymbol{\pi}^0}_{b,m,G}, k^{\boldsymbol{\pi}^1}_{b,m,G}\right)$.
\State $\rho^*(b,m,G) = \arg\min_i k^{\boldsymbol{\pi}^i}_{b,m,G}$ for $i = 0,1$
\EndFor
\State $n\gets n+1$
\EndFor
\end{algorithmic}
\end{algorithm}

\section{Numerical Results}\label{Results}

In this section, we provide numerical evidence to support the analytical results established in the paper. VIA is a standard tool for solving the bellman equations in (\ref{actionvalue}). However, VIA iterates for numerous passes over each state, which is increasing in $\beta$, before converging to a steady solution, whereas Algorithm \ref{alg1} and \ref{alg2} needs a single iteration. Moreover, VIA  achieves exactly the same performance as Algorithm \ref{alg1} and \ref{alg2}. Thus, we omit the results obtained by VIA.

We will divide our attention to validate the optimal policy for i.i.d. and correlated channel  models. Although the framework discussed is sufficiently general to determine the number of re-transmissions starting from any residual battery level, in this section for the clarity of presentation, we consider that the initial battery level is zero. We use a simple ARQ mechanism as a baseline for understanding the performance merits of the HARQ mechanism. In the following, we formally define the simple ARQ scheme for i.i.d. and correlated channels.

\subsection{Simple ARQ}\label{simple_harq}
In simple ARQ, the packet is transmitted successfully whenever the channel is in a GOOD state and the receiver has sufficient energy to decode the packet. Otherwise, the receiver drops the packet and awaits re-transmissions. When the receiver employs simple ARQ, before any decoding attempt, it has to make sure that its battery has at least $E_d+1$ units of energy. Otherwise, after consuming 1 unit of energy for sampling, it will not have sufficient energy to decode the data packet and it will drop the packet. It is easy to prove that the optimal simple ARQ policy minimizing the mean time to absorption first harvests $E_d+1$ units of energy and then attempts decoding. If the decoding attempt is not successful, it harvests energy until its battery state reaches $E_d+1$ units again before attempting to decode. 



\subsection{i.i.d. Channel States}
\label{num_iid}
In this section, we evaluate the minimum mean time to absorption obtained from Algorithm 1, and compare it to that of the following three simple policies. The studied policies are as follows: $i)$ Battery First (BF),  $ii)$ Information First (IF), and $iii)$ Coin Toss (CT). Also, we compare the performance of the receiver equipped with HARQ mechanism with the case of a receiver equipped with simple ARQ mechanism. We determine the mean number of re-transmissions by Monte Carlo simulations, and compare them with that of analytical calculation described in Algorithm 1.  Note that Monte Carlo simulations provide only \emph{sample} mean time which is a random variable. The  mean of this random variable is equal to the mean time to absorption and its variance decreases with the number of samples and becomes zero only if the number of iterations go to infinity. Hence, we expect to see small differences between the results obtained by the Monte Carlo simulations and analytical results, which is the reason why some policies have slightly smaller mean time to absorption than the optimal analytical value.

Table \ref{son} summarizes the mean time to absorption for $R_1 = 10$, $e=1$, $\lambda=0.5$ and $E_d=5$ with respect to $R_0$ associated with different policies. For IF, BF, CT and simple ARQ policies, we run Monte Carlo simulations for $10^7$ iterations and evaluate the sample mean. It can be seen from Table \ref{son} that all policies have almost the same performance. This observation confirms our major finding that the optimal policy achieving the minimum mean time to absorption is not unique. 

The effect of quality of the channel on the mean time to absorption for $R_0 = 5$, $R_1=10$, $E_d=5$ and $e=2$ with respect to $\lambda$ is summarized in Table \ref{son_vs_lamda}. As expected, it can be seen that the mean time to absorption decreases as the channel quality improves. Also, the performance gap between the HARQ and simple ARQ mechanism becomes smaller as the channel quality improves. This is because as the channel quality improves, the probability of harvesting energy and accumulating $R_1$ bits of mutual information also increases. Finally, the mean time to absorption for $R_0 = 5$, $R_1=10$, $E_d=10$ and $\lambda=0.3$ with respect to $e$ is summarized in Table \ref{son_vs_l}. We observe that the mean time to absorption is approximately the same for all policies and it is decreasing with respect to the amount of harvested energy, $e$.

The results presented in Table \ref{son}, \ref{son_vs_lamda} and \ref{son_vs_l} confirm our theoretical results that, indeed, the optimal policy harvests energy whenever $b = 0$ or $m=R_1$ and accumulates mutual information whenever $b>E_d$. For the rest of the states it does not matter what the receiver does, as long as, it does not split the received RF signal.

\begin{table*}
\centering
\caption{Mean time to absorption for $R_1 = 10$, $e=1$ and $E_d=5$ vs. $R_0$}
\label{son}
\begin{tabular}{lllllllllll}
               & $R_0=1$  & $R_0=2$  & $R_0=3$  & $R_0=4$  & $R_0=5$  & $R_0=6$   & $R_0=7$   & $R_0=8$   & $R_0=9$     \\
Optimal analytical & 15.9941 & 15.8125 & 15.6250 & 15.2500 & 14.5000 & 14.5000 & 14.5000 & 14.5000 & 14.5000  \\
Optimal Monte-Carlo & 15.9910 & 15.8103 & 15.6235 & 15.2490 & 14.4992 & 14.5001 & 14.4998 & 14.5000 & 14.4983  \\
BF             & 15.9938 & 15.8116 & 15.6259 & 15.2504 & 14.4999 & 14.4995 & 14.4993 & 14.5012 & 14.5000  \\
IF             & 15.9941 & 15.8143 & 15.6245 & 15.2508 & 14.4987 & 14.4997 & 14.5017 & 14.4989 & 14.5003  \\
CT             & 15.9966 & 15.8140 & 15.6266 & 15.2491 & 14.5020 & 14.5007 & 14.5009 & 14.4984 & 14.5001 \\
Simple ARQ     & 15.9992 & 15.9992 & 15.9992 & 16.0006 & 16.0007 & 15.9995 & 15.9996 & 16.0008 & 16.0011 
\end{tabular}
\end{table*}

\begin{table*}
\centering
\caption{Mean time to absorption for $R_1 = 10$, $R_0=5$, $e=2$ and $E_d=5$ vs. $\lambda$}
\label{son_vs_lamda}
\begin{tabular}{lllllllllll}
               & $\lambda=0.1$  & $\lambda=0.2$  & $\lambda=0.3$  & $\lambda=0.4$  & $\lambda=0.5$  & $\lambda=0.6$   & $\lambda=0.7$   & $\lambda=0.8$   & $\lambda=0.9$     \\
Optimal analytical & 40.9000 & 20.8000 & 14.0333 & 10.6000 & 8.5000 & 7.0667 & 6.0143 & 5.2000 & 4.5444  \\
Optimal Monte-Carlo & 40.8904 & 20.7979 & 14.0320 & 10.5985 & 8.4989 & 7.0659 & 6.0140 & 5.1999 & 4.5443  \\
BF             & 40.8920 & 20.7962 & 14.0337 & 10.6002 & 8.4995 & 7.0666 & 6.0153 & 5.1998 & 4.5445  \\
IF             & 40.8978 & 20.7960 & 14.0331 & 10.5991 & 8.5002 & 7.0667 & 6.0132 & 5.1998 & 4.5443  \\
CT             & 40.8961 & 20.8006 & 14.0333 & 10.5973 & 8.4986 & 7.0665 & 6.0137 & 5.2001 & 4.5444 \\
Simple ARQ     & 87.3286 & 31.1145 & 17.9077 & 12.3428 & 9.3310 & 7.4591 & 6.1846 & 5.2607 & 4.5568 
\end{tabular}
\end{table*}

\begin{table*}
\centering
\caption{Mean time to absorption for $R_1 = 10$, $R_0=5$, $\lambda=0.3$ and $E_d=10$ vs. $e$}
\label{son_vs_l}
\begin{tabular}{lllllllllll}
               & $e=1$  & $e=2$  & $e=3$  & $e=4$  & $e=5$  & $e=6$   & $e=7$   & $e=8$   & $e=9$     \\
Optimal analytical & 40.7000 & 21.7000 & 15.0333 & 11.7000 & 11.7000 & 8.3667 & 8.3667 & 8.3667 & 8.3667  \\
Optimal Monte-Carlo & 40.6956 & 21.6999 & 15.0320 & 11.6995 & 11.7009 & 8.3648 & 8.3651 & 8.3675 & 8.3670  \\
BF             & 40.7015 & 21.6987 & 15.0381 & 11.6986 & 11.6995 & 8.3677 & 8.3653 & 8.3667 & 8.3672  \\
IF             & 40.7023 & 21.7020 & 15.0308 & 11.7030 & 11.7010 & 8.3667 & 8.3658 & 8.3671 & 8.3654  \\
CT             & 40.6980 & 21.7006 & 15.0345 & 11.6995 & 11.6992 & 8.3674 & 8.3663 & 8.3657 & 8.3670 \\
Simple ARQ     & 47.7832 & 26.5340 & 19.1515 & 15.4839 & 14.0076 & 11.8479 & 10.8730 & 10.4191 & 10.2021
\end{tabular}
\end{table*}

\subsection{Correlated Channel}
In this section, we investigate the performance of the optimal policy presented in Algorithm \ref{alg2} for the case of correlated channel and compare its performance to the three baseline policies that employ HARQ mechanism as well as a simple ARQ mechanism. We also consider a randomized policy, which we call  \emph{Bernoulli} policy which harvests energy with probability, $p$, unless its battery state is less than one unit or it has accumulated sufficient mutual information during when it solely harvests energy. In the following, we study the effects of the encoding rate, the time correlation, and the EH rate.  Note that the mean time to absorption is determined by calculating $k_{b,m,0}$ and $k_{b,m,1}$ and then averaging them with respect to the steady-state distribution of the channel states, i.e., $k_{b,m} = \phi(0)k_{b,m,0} + \phi(0)k_{b,m,1}$, where $\phi(0) =1-\phi(1)= \frac{1-\lambda_1}{1+\lambda_0-\lambda_1}$.

\begin{remark}
Note that, in this section, we do not calculate the mean time to absorption by Algorithm \ref{alg2} (i.e., $k^{\boldsymbol{\pi}^*}_{b,m,0}$ and $k^{\boldsymbol{\pi}^*}_{b,m,1}$). Instead, we use the optimal scheduling decisions dictated by Algorithm \ref{alg2}  for each state $(b,m, G)$ to determine the mean time to absorption by Monte-Carlo simulations. This is because both methods yield the same mean time to absorption for the optimal policy and illustrating both on the same figure distinctly is not possible.
\end{remark}

To investigate the effect of the encoding rate on the mean time to absorption, we set the simulation parameters as $R_1 = 10$, $e=1$, $E_d=5$ and $p=0.1$. The mean time to absorption with respect to $R_0$, for negatively and positively correlated channel states, are depicted in Figures \ref{cor_VS_R1} and \ref{cor_VS_R2}, respectively. 
Unlike the i.i.d. case the knowledge of the channel state makes a significant difference in the performance of the proposed optimal policy as compared to the baseline policies. Hence, when the channel is correlated, a simple scheduling policy is not sufficient to achieve a low number of re-transmissions.

\begin{figure}
        \centering
        \begin{subfigure}[t]{0.4\textwidth}
                \includegraphics[width=\textwidth]{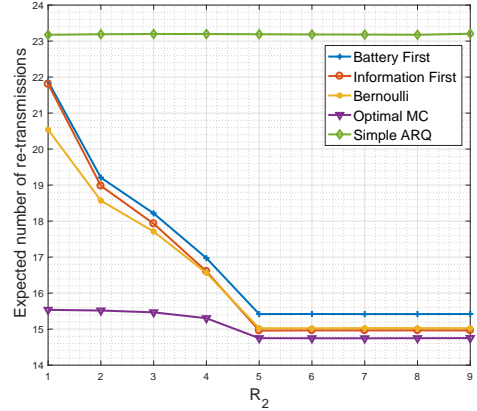}
                \caption{Negatively correlated channel, $\lambda_0 =0.7$ and $\lambda_1 =0.2$.}
                \label{cor_VS_R1}
        \end{subfigure}
        \hfill
        \begin{subfigure}[t]{0.4\textwidth}
                \includegraphics[width=\textwidth]{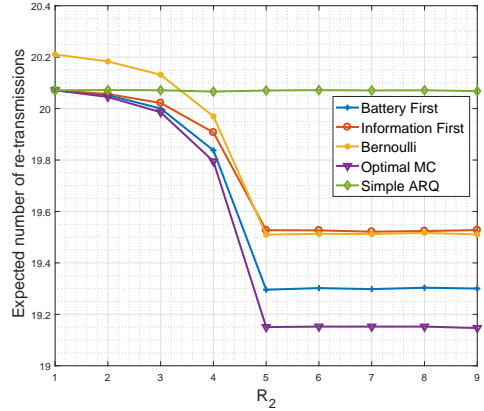}
                \caption{Positively correlated channel $\lambda_0 =0.2$ and $\lambda_1 =0.7$. }
                \label{cor_VS_R2}
        \end{subfigure}
				\caption{The effect of the encoding rate on the  minimum expected number of re-transmissions for $R_1 = 10$, $e=1$, $E_d=5$ and $p=0.1$.}
\label{cor_VS_R}
\end{figure}

Next, we study the effect of the channel quality and the correlation on the mean time to absorption. We set $R_1 = 10$, $e=1$, $E_d=5$ and $p=0.1$. We fix $\lambda_1=0.2$ and by varying $\lambda_0$, we calculate the mean time to absorption as illustrated in Figure \ref{cor_VS_lamda1}. Similarly, we fix $\lambda_0=0.2$ and by varying $\lambda_1$, we calculate the mean time to absorption by the aforementioned baseline policies  and illustrate the results in Figure \ref{cor_VS_lamda2}. Note that when the channel is negatively correlated, as in Figure \ref{cor_VS_lamda2}, the gap between the optimal policy and the baseline policies is high. However, when the channel is positively correlated, as in Figure \ref{cor_VS_lamda2}, the gap disappears as $\lambda_1$ increases. This is because, when the channel is positively correlated, the channel  tends to stay in the same state for a longer time before changing its state. On the contrary, in negatively correlated channel states, the channel is more likely to change its state at any time. This rapid change in state transition in the case of negatively correlated channel states requires a more adaptive policy rather than the case of the positively correlated channel state which rarely changes its state. Thus, the performance gain of Algorithm \ref{alg2} is more evident in negatively correlated channels.

\begin{figure}
        \centering
        \begin{subfigure}[t]{0.4\textwidth}
                \includegraphics[width=\textwidth]{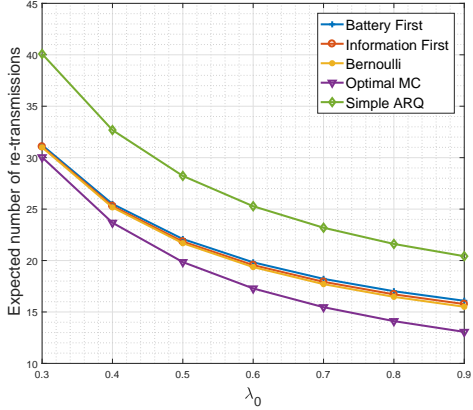}
                \caption{Negatively correlated channel, $\lambda_1 =0.2$.}
                \label{cor_VS_lamda1}
        \end{subfigure}
        \hfill
        \begin{subfigure}[t]{0.4\textwidth}
                \includegraphics[width=\textwidth]{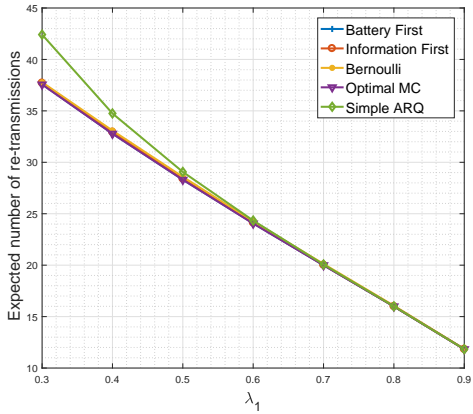}
                \caption{Positively correlated channel, $\lambda_0 =0.2$. }
                \label{cor_VS_lamda2}
        \end{subfigure}
				\caption{The effect of the channel quality and correlation on  the minimum expected number of re-transmissions for $R_1 = 10$, $R_0 = 3$, $e=1$, $E_d=5$ and $p=0.1$.}
\label{cor_VS_lamda}
\end{figure}

Finally the effect of EH rate, $e$, on the mean time to absorption for negatively and positively correlated channel states is depicted in Figure \ref{cor_VS_l1} and \ref{cor_VS_l2}, respectively. The results are obtained by setting $R_1 = 10$, $R_0=5$, $E_d=10$, $p=0.1$, $\lambda_0 =0.7$ and $\lambda_1 =0.2$ for negatively correlated channel states; and $\lambda_0 =0.2$ and $\lambda_1 =0.7$ for positively correlated channel states. We, again, observe that the optimal policy outperforms the baseline policies and the performance gain is more evident for negatively correlated channel states for the same reason we provided for the results in Figure \ref{cor_VS_lamda}.

It should be noted that when the channel states are correlated, the knowledge about the future channel states plays a major role in making decision about the power splitting ratio. On the contrary, when the channel states evolve i.i.d. over time, there exist a class of optimal policies instead of a single optimal policy.

\begin{figure}
        \centering
        \begin{subfigure}[t]{0.4\textwidth}
                \includegraphics[width=\textwidth]{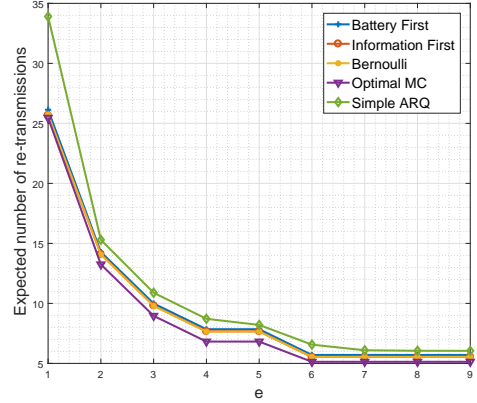}
                \caption{$\lambda_0 =0.7$ and $\lambda_1 =0.2$.}
                \label{cor_VS_l1}
        \end{subfigure}
        \hfill
        \begin{subfigure}[t]{0.4\textwidth}
                \includegraphics[width=\textwidth]{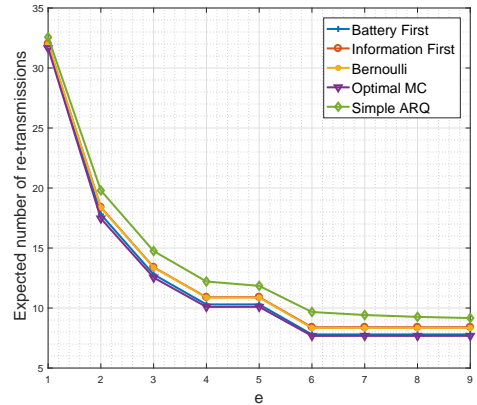}
                \caption{$\lambda_0 =0.2$ and $\lambda_1 =0.7$. }
                \label{cor_VS_l2}
        \end{subfigure}
				\caption{The effect of the EH rate on the minimum expected number of re-transmissions for $R_1 = 10$, $R_0=5$, $E_d=10$ and $p=0.1$.}
\label{cor_VS_l}
\end{figure}

\section{Conclusion}\label{Conclusions}
We analyzed a point-to-point wireless link employing HARQ for reliable transmission, where the receiver can only empower itself via the transmitter's RF signal. We modeled the problem of optimal power splitting using a Markovian framework, and developed an optimal algorithm achieving the minimum mean time to absorption for both time varying i.i.d. and correlated channels. We developed computationally inexpensive algorithms to calculate the minimum mean time to absorption and optimize the power splitting ratio starting at any arbitrary state. 

We proved that the optimal policy in case of i.i.d. channel states is not unique, and indeed the optimal policy belongs to the optimal family of policies. For correlated channel, we observed that it is only possible to achieve the optimal performance by intelligently utilizing the information offered by channel's correlation information. Finally, we numerically validated the analytical results established in the paper by providing extensive number of simulations.

It is worth mentioning that the two-state model, adopted here, is an approximation of a more general multi-state wireless channel. As a future work, we aim to extend this work for a more general setting where we will consider multi-rate information transmission, multi-state EH process, and non-linear EH efficiency. Due to  to analytical complexity, it is uncertain that the optimality result of no-split policy carries over to the more general setting. In this case, deep reinforcement learning techniques can be used as a promising approach to address the aforementioned extensions.



\begin{appendices}
\section{Proof of Lemma \ref{lemma1}}
\label{prooflemma1}
The proof is by induction. 
\begin{enumerate}
\item Base case: Let us consider the smallest possible value for $i$, i.e., $i=1$, such that $E_d-e\leq b < E_d$. Note that since $m=R_1$, the optimal decision is to use incoming RF signal only for harvesting energy, i.e., $\rho^*(b,\ R_1)=1$. Thus, we get
\small
\begin{align}
k_{b,R_1} = 1 + \lambda k_{b+e,R_1}+ (1-\lambda) k_{b,R_1}.
\end{align}
\normalsize
For $E_d-e\leq b < E_d$, if the channel is GOOD then the MC transitions into state $(b+e,\ R_1)$, which is an absorbing state,  so $k_{b+e,R_1}=0$. Hence, $k_{b,R_1}=\frac{1}{\lambda}$ and thus, the lemma holds for $i = 1$.
\item Induction step: assume that the lemma is true for some $i=n$, i.e., $k_{b,R_1}=n/\lambda$ for $E_d-n\cdot e\leq b<E_d-(n-1)\cdot e$.
\item Proof for case $i = n+1$: Let us calculate the mean time to absorption for the case $n+1$:
\small
\begin{align}
k_{b,R_1} =& 1 + \lambda k_{b+e,R_1}+ (1-\lambda)k_{b,R_1},\nonumber \\  &\hspace{2.5cm}\text{for}\ \ E_d-(n+1)e\leq b<E_d-nl,
\end{align}
\normalsize
which reduces to $k_{b,R_1} = \frac{n+1}{\lambda}$ for $E_d-(n+1)\cdot e\leq b<E_d-n\cdot e$.
\end{enumerate}
Thus, the lemma holds by induction.

\section{Proof of Theorem \ref{TH_dis}}
\label{proofTH_dis}
Assume that at time slot $t$ the system is at state $(b,\ m)$. Consider policy $\boldsymbol{\pi}^{split}$ which always chooses $0< \rho<1$. Hence, it follows that $R^H(\rho)< R_1$, $R^L(\rho)< R_0$ and, from (\ref{I}), we have $I(t)\leq R_1$. Also, it is easy to verify that for any $b$, we have $k_{b,m_1}\leq k_{b,m_2}$ whenever $m_1\geq m_2$. Thus, a lower bound on $k^{\boldsymbol{\pi}^{split}}_{b,m}$ in (\ref{k_pi}) can be established as,
\small
\begin{align}
k^{\boldsymbol{\pi}^{split}}_{b,m}\geq 1 + \lambda k_{b-1+\rho e,R_1} + (1-\lambda) k_{b-1,m+R_0}.\label{eq:bound}
\end{align}
\normalsize
Furthermore, since $b-1<b-1+\rho\cdot e<b-1+e$, from Lemma \ref{lemma1}, we know that $k_{b-1+\rho e,R_1}=k_{b-1,R_1}$. Hence, the lower bound in (\ref{eq:bound}) is exactly the same as $k^{\boldsymbol{\pi}^{no-split}}_{b,m}$ given in (\ref{kalpha}), i.e., $k^{\boldsymbol{\pi}^{no-split}}_{b,m}\leq k^{\boldsymbol{\pi}^{split}}_{b,m}$.

\section{Proof of Theorem \ref{lemma2}}
\label{prooflemma2}
The proof is by induction. For the base case consider the initial case when $j = 1$ so that $b= E_d+1, E_d+2, \ldots$ and $m = R_1-1$. We have 
\small
\begin{align}
k^{\boldsymbol{\pi}^0}_{E_d+1,R_1-1} =& 1 + \lambda k^{\boldsymbol{\pi}^*}_{E_d,R_1} + (1-\lambda)k^{\boldsymbol{\pi}^*}_{E_d,R_1}= 1,\\
k^{\boldsymbol{\pi}^1}_{E_d+1,R_1-1} =& \frac{1}{\lambda} + k^{\boldsymbol{\pi}^*}_{E_d+e+1,m}\nonumber\\
>&k^{\boldsymbol{\pi}^0}_{E_d+1,R_1-1}.
\end{align}
\normalsize
Note that when $b= E_d+1, E_d+2, \ldots$, by choosing $\rho =0$, regardless of the channel state, the next state, $(b-1,R_1)$, is an absorbing state so $k^{\boldsymbol{\pi}^0}_{b,R_1-1}=1$. Thus, the lemma holds for $j = 1$. In the induction step assume that the theorem holds for $j=n-1$, i.e., $k^{\boldsymbol{\pi}^*}_{b,R_1-n+1}=k^{\boldsymbol{\pi}^0}_{E_d+n-1,R_1-n+1}=\sum^{n-1}_{i=1} (1-\lambda)^{i-1}$ for $b= E_d+n-1, E_d+n, \ldots$. Now, we prove that the claim is also true for $j=n$.
\small
\begin{align}
k^{\boldsymbol{\pi}^0}_{E_d+n,R_1-n} =& 1 + (1-\lambda)k^{\boldsymbol{\pi}^*}_{E_d+n-1,R_1-n+1}\nonumber\\
=& 1 + \sum^{n-1}_{i=1} (1-\lambda)^{i}\nonumber\\
=& \sum^{n}_{i=1} (1-\lambda)^{i-1},\\
k^{\boldsymbol{\pi}^1}_{E_d+n,R_1-n}=& \frac{1}{\lambda} + k^{\boldsymbol{\pi}^*}_{E_d+n+e,R_1-n}\nonumber\\
>&\frac{1}{\lambda} + k^{\boldsymbol{\pi}^*}_{E_d+n+e,R_1-n+1}\nonumber\\
=& \frac{1}{\lambda} + k^{\boldsymbol{\pi}^*}_{E_d+n-1,R_1-n+1}\nonumber\\
=&\frac{1}{\lambda} + \frac{1-(1-\lambda)^{n-1}}{\lambda}
\end{align}
\normalsize
Furthermore,
\small
\begin{align}
k^{\boldsymbol{\pi}^0}_{E_d+n,R_1-n} =&\frac{1-(1-\lambda)^n}{\lambda}\nonumber\\
=& 1 + (1-\lambda)\frac{1-(1-\lambda)^{n-1}}{\lambda}<k^{\boldsymbol{\pi}^1}_{E_d+n,R_1-n}
\end{align}
\normalsize
For the last part of the proof, we need to show that $k^{\boldsymbol{\pi}^*}_{b,R_1-n} = k^{\boldsymbol{\pi}^0}_{E_d+n,R_1-n}$ for $b= E_d+n+1, E_d+n+2, \ldots$. We may write:
\small
\begin{align}
k^{\boldsymbol{\pi}^*}_{b,R_1-n} &= 1 + (1-\lambda)k^{\boldsymbol{\pi}^0}_{b-1,R_1-n+1}\nonumber\\
&= 1 + (1-\lambda)k^{\boldsymbol{\pi}^0}_{E_d+n-1,R_1-n+1} = k^{\boldsymbol{\pi}^0}_{E_d+n,R_1-n}
\end{align}
\normalsize

\section{ The Optimality of Algorithm \ref{alg1}}
\label{proofalg1}
In Lemma \ref{lemma1}, we characterized the minimum mean time to absorption for all states $(b,\ R_1)$, for $b=0,\ldots,E_d-1$. Also, in Theorem \ref{lemma2}, we characterized the minimum mean time to absorption for states, $(b,\ R_1-j)$ where, $b=E_d+j,E_d+j+1,\ldots$ and $j = 1,\ldots,R_1$. Furthermore, Theorem \ref{TH_dis} proves that at any state $(b,\ m)$, the receiver should either choose to harvest energy or accumulate mutual information. Note that the iterations are ordered in Algorithm \ref{alg1} (line 4-8) so that $k^{\boldsymbol{\pi}^0}_{b,m}$ and $k^{\boldsymbol{\pi}^1}_{b,m}$ only depend on $k^{\boldsymbol{\pi}^*}_{b-1,R_1}$,  $k^{\boldsymbol{\pi}^*}_{b-1,m+1}$, and $k^{\boldsymbol{\pi}^*}_{b+1,m}$ which are obtained at the previous rounds of the algorithm.

 \section{The Optimality of $\rho=0$ when $b>E_d$}
 \label{proofTH_str1}
 Due to space limitations, we only provide a sketch of the proof with complete details given in \cite{arxivmehdi}. We need to show that $k^{\boldsymbol{\pi}^0}_{E_d+j-i,R_1-j}<k^{\boldsymbol{\pi}^1}_{E_d+j-i,R_1-j}$ for all $j = 1,\ldots,R_1$ and $i=0,1,\ldots,j-1$. The proof is by induction. For the base case, we need to show that the theorem holds for $i=0$ and all $j = 1,\ldots,R_1$, i.e., $k^{\boldsymbol{\pi}^0}_{E_d+j,R_1-j}<k^{\boldsymbol{\pi}^1}_{E_d+j,R_1-j}$, which is an immediate result of Theorem \ref{lemma2}. Next, in the induction step, assume that the theorem is true for $i=n$ and all $j = 1,\ldots,R_1$ i.e., $k^{\boldsymbol{\pi}^0}_{E_d+j-n,R_1-j}<k^{\boldsymbol{\pi}^1}_{E_d+j-n,R_1-j}$. Then, using (\ref{k0}) and (\ref{k1}), it is possible to show that:
\small
\begin{align}
&k^{\boldsymbol{\pi}^1}_{E_d+j-(n+1),R_1-j}=\frac{1}{\lambda} + 1+(1-\lambda)k^{\boldsymbol{\pi}^*}_{E_d+(j-1)-n,R_1-(j-1)},\label{k1-TH_str1}\\
&k^{\boldsymbol{\pi}^0}_{E_d+j-(n+1),R_1-j}\leq \frac{1}{\lambda} +(1-\lambda)k^{\boldsymbol{\pi}^*}_{E_d+(j-1)-n,R_1-(j-1)},
\end{align} 
\normalsize
which results in $k^{\boldsymbol{\pi}^0}_{E_d+j-(n+1),R_1-j} < k^{\boldsymbol{\pi}^1}_{E_d+j-(n+1),R_1-j}$, proving that the statement also holds for $i=n+1$, and all $j = 1,\ldots,R_1$.
 \section{ $k^{\boldsymbol{\pi}^0}_{b,m} = k^{\boldsymbol{\pi}^1}_{b,m}$ for $1\leq b\leq E_d$, $0\leq m\leq R_1-1$}
 \label{proofTH_str2}
 Due to space limitations, we only provide a sketch of the proof with complete details given in \cite{arxivmehdi}. We have to show that $k^{\boldsymbol{\pi}^0}_{i,R_1-j} = k^{\boldsymbol{\pi}^1}_{i,R_1-j}$ for $i=1,\ldots,E_d$ and $j = 1,\ldots,R_1$. The outline of the induction proof is as follows:
\begin{itemize}
\item For the base case we show that $k^{\boldsymbol{\pi}^0}_{i,R_1-1} = k^{\boldsymbol{\pi}^1}_{i,R_1-1}$ for all $i=1,\ldots,E_d$. It is easy to verify that $k^{\boldsymbol{\pi}^0}_{E_d,R_1-1} = k^{\boldsymbol{\pi}^1}_{E_d,R_1-1}$. By assuming that $k^{\boldsymbol{\pi}^0}_{i,R_1-1} = k^{\boldsymbol{\pi}^1}_{i,R_1-1}$, from (\ref{k0}) and (\ref{k1}) , one can calculate:
\begin{align}
k^{\boldsymbol{\pi}^1}_{i-1,R_1-1}=k^{\boldsymbol{\pi}^0}_{i-1,R_1-1}=1+\frac{E_d-i+2}{\lambda}.
\end{align}

\item In the induction step, we assume that the theorem is true for $j=n$ and all $i=1,\ldots,E_d$.
\item Using the induction step, (\ref{k0}), and (\ref{k1}), we obtain the following result for the case $n+1$:
\begin{align}
&k^{\boldsymbol{\pi}^0}_{i-1,R_1-(n+1)}\nonumber\\
&=k^{\boldsymbol{\pi}^1}_{i-1,R_1-(n+1)}=\frac{1}{\lambda} + E_d-i+2 + (1-\lambda) k^{\boldsymbol{\pi}^*}_{i-1,R_1-n}.
\end{align}
\end{itemize}
Hence, the theorem holds for $j = n+1$ and all $i=1,\ldots,E_d$. Therefore, the statement is true by induction.

\section{Proof of Theorem \ref{TH_str}}
\label{proofTH_str}
The proof of the theorem is straightforward and proceeds as follows:
\begin{enumerate}
\item When $b=0$, the receiver has no energy to activate the RF transceiver and should first recharge its battery. When $m=R_1$, the receiver collected sufficient mutual information to decode, but needs energy to perform the decoding operation.  Hence, it harvests energy.
\item This part of the theorem is proven in Appendix \ref{proofTH_str1}.
\item In Appendix \ref{proofTH_str2} we show that whenever $b=1,2,\ldots, E_d$, and $m=0,1,\ldots R_1-1$, then $k^{\boldsymbol{\pi}^0}_{b,m} = k^{\boldsymbol{\pi}^1}_{b,m}$. Consider a policy $\beta$ which satisfies part 1 and 2 of the theorem. Whenever $b=1,2,\ldots, E_d$ and $m=0,1,\ldots R_1-1$, the policy chooses $\rho=0$ with probability $p$. The mean time to absorption of policy $\beta$, $k^{\beta}_{b,m}$ can be calculated as follows
\end{enumerate}
\begin{align}
k^{\beta}_{b,m} = p k^{\boldsymbol{\pi}^0}_{b,m} + (1-p) k^{\boldsymbol{\pi}^1}_{b,m} =  k^{\boldsymbol{\pi}^0}_{b,m}=k^{\boldsymbol{\pi}^1}_{b,m}.
\end{align}
\end{appendices}
\bibliographystyle{IEEEtran}
\bibliography{ref.bib}

\begin{IEEEbiography}[{\includegraphics[width=1in,height=1.25in,clip,keepaspectratio]{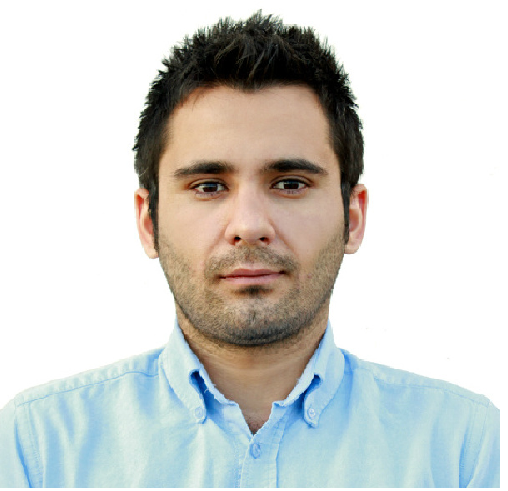}}]{Mehdi Salehi Heydar Abad}
received the B.S. degree in electrical engineering from the IUST, Tehran, Iran, in 2012, and the M.S. degree in electrical and electronics engineering from the Sabanci University, Istanbul, Turkey in 2015. Currently he is a PhD candidate at Sabanci University. He was also a Visiting Researcher at The Ohio State University, Columbus, OH, USA. His research interests are in the field of mathematical modeling of communication systems, stochastic optimization, and green communication networks. He is the recipient of the Best Paper Award at the 2016 IEEE Wireless Communications and Networking Conference (WCNC).
\end{IEEEbiography}

\begin{IEEEbiography}[{\includegraphics[width=1in,height=1.25in,clip,keepaspectratio]{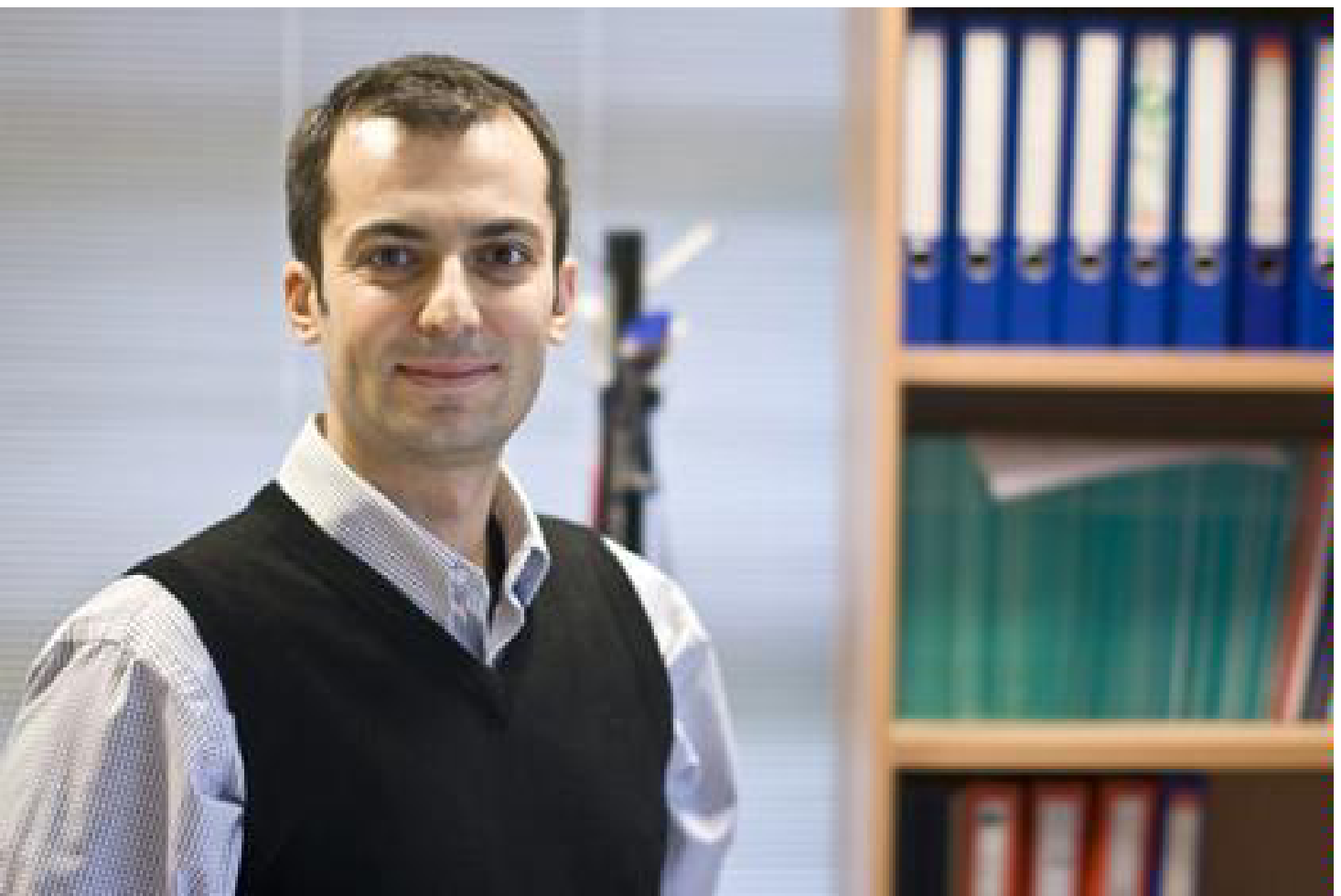}}]{Ozgur Ercetin}
received the B.S. degree in electrical and electronics engineering from the Middle East Technical University, Ankara, Turkey, in 1995, and the M.S. and Ph.D. degrees in electrical engineering from the University of Maryland, College Park, MD, USA, in 1998 and 2002, respectively. Since 2002, he has been with the Faculty of Engineering and Natural Sciences, Sabanci University, Istanbul, Turkey. He was also a Visiting Researcher with HRL Labs, Malibu, CA, USA; Docomo USA Labs, Palo Alto, CA, USA; The Ohio State University, Columbus, OH, USA; Carleton University, Ottawa, CA and Université du Québec à Montréal, Montreal CA. His research interests are in the field of computer and communication networks with emphasis on fundamental mathematical models, architectures and protocols of wireless systems, and stochastic optimization.
\end{IEEEbiography}

\begin{IEEEbiography}[{\includegraphics[width=1in,height=1.25in,clip,keepaspectratio]{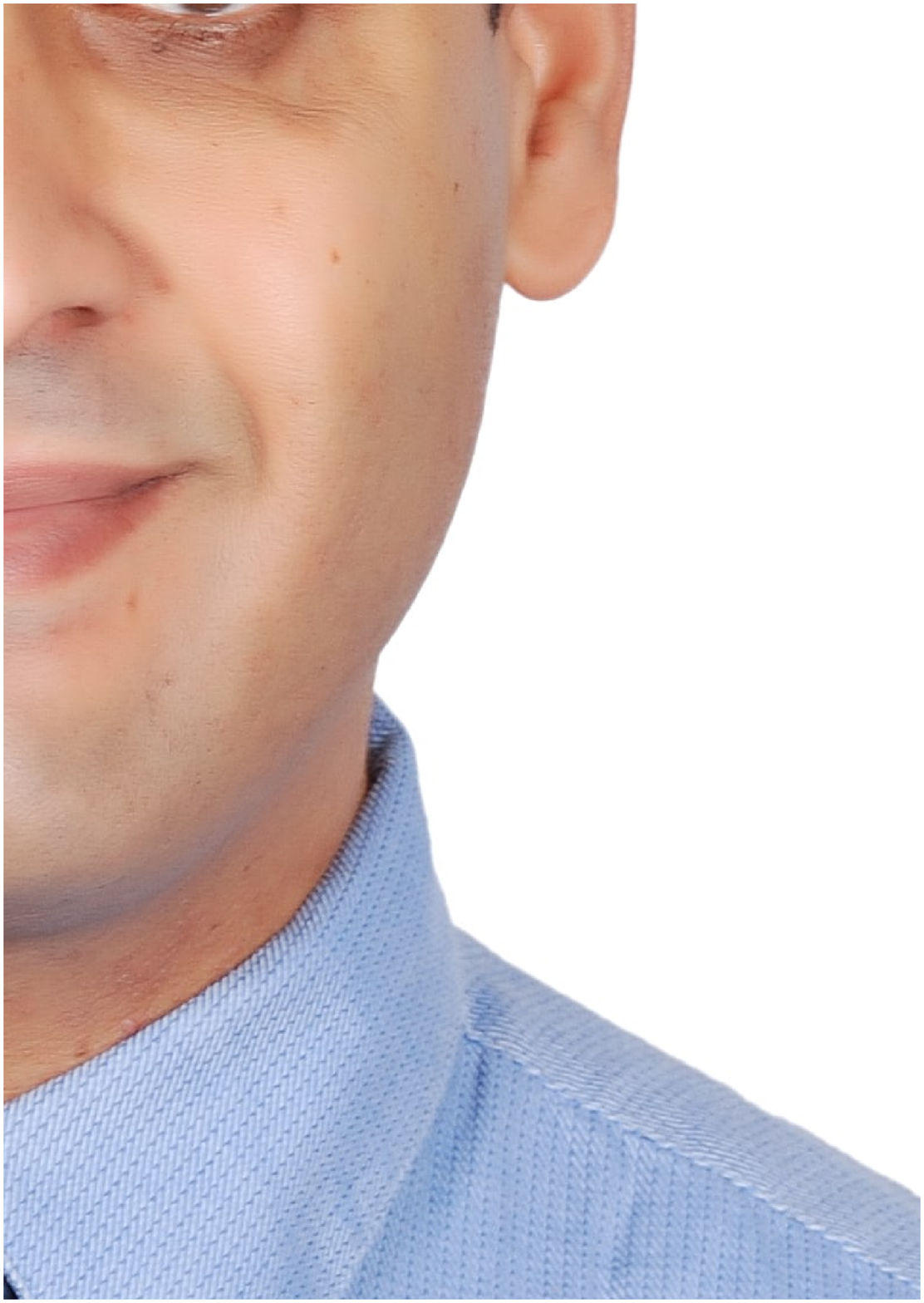}}]{Tamer Elbatt}
[S’98-M’01-SM’06] received the B.S. and M.S. degrees in EECE from Cairo University, Egypt in 1993 and 1996, respectively, and the Ph.D. degree in Electrical and Computer Engineering from the University of Maryland, College Park, MD, USA in 2000. From 2000 to 2009 he was with major U.S. industry R\&D labs, e.g., HRL Laboratories, LLC, Malibu, CA, USA and Lockheed Martin ATC, Palo Alto, CA, USA, at various positions. From 2009 to 2017, he served at the EECE Dept., Faculty of Engineering, Cairo University, as an Assistant Professor and later as an Associate Professor, currently on leave. He also held a joint appointment with Nile University, Egypt from 2009 to 2017 and has served as the Director of the Wireless Intelligent Networks Center (WINC) from 2012 to 2017. In July 2017, he joined the Dept. of CSE at the American University in Cairo as an Associate Professor. Dr. ElBatt research has been supported by the U.S. DARPA, ITIDA, QNRF, EU FP7, H2020, General Motors, Microsoft, Google and Vodafone Egypt Foundation and is currently being supported by Egypt NTRA. He has published more than 110 papers in prestigious journals and international conferences. Dr. ElBatt holds seven issued U.S. patents and one WIPO patent.
Dr. ElBatt is a Senior Member of the IEEE and has served on the TPC of numerous IEEE and ACM conferences. He served as the Demos Co-Chair of ACM Mobicom 2013 and the Publications Co-Chair of IEEE Globecom 2012 and EAI Mobiquitous 2014. Dr. ElBatt currently serves on the Editorial Board of IEEE Transactions on Cognitive Communications and Networking and Wiley International Journal of Satellite Communications and Networking and has previously served on the Editorial Board of IEEE Transactions on Mobile Computing. Dr. ElBatt has also served on the United States NSF and Fulbright review panels. Dr. ElBatt was a Visiting Professor at the Dept. of Electronics, Politecnico di Torino, Italy in Aug. 2010, Faculty of Engineering and Natural Sciences, Sabanci University, Turkey in Aug. 2013 and the Dept. of Information Engineering, University of Padova, Italy in Aug. 2015. Dr. ElBatt is the recipient of the prestigious Google Faculty Research Award in 2011, the 2012 Cairo University Incentive Award in Engineering Sciences and the 2014 Egypt’s State Incentive Award in Engineering Sciences. His research interests lie in the broad areas of performance analysis, design and optimization of wireless and mobile networks.
\end{IEEEbiography}

\begin{IEEEbiography}[{\includegraphics[width=1in,height=1.25in,clip,keepaspectratio]{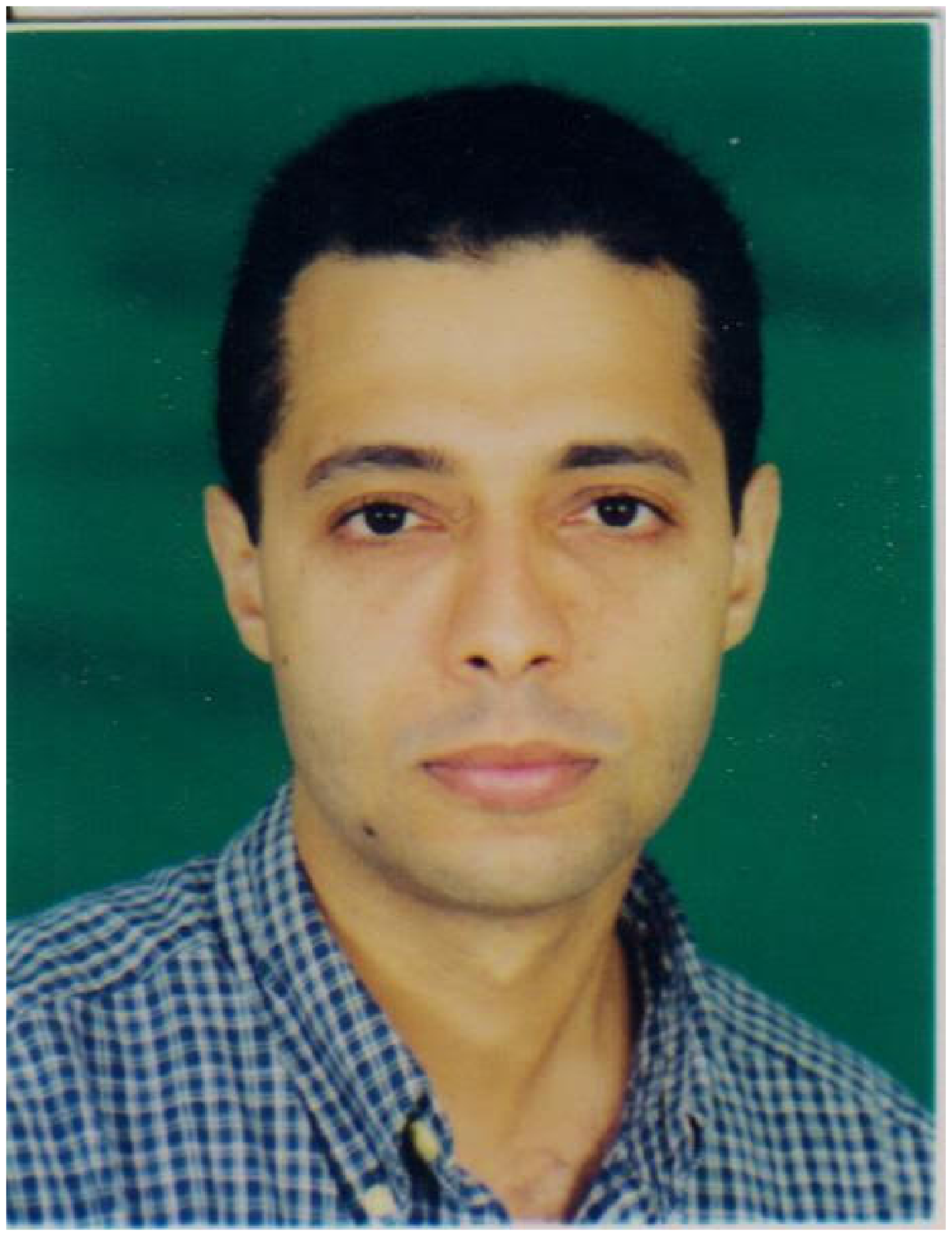}}]{Mohammed Nafie}
[S’03-M’08-SM’13] received his B.Sc. degree with honors from Cairo University in 1991. He received his M.Sc. degree with Distinction from King's College, University of London in 1993. He received his Ph.D. degree from University of Minnesota in 1999 where his research interests were in the area of low complexity detection algorithms. After his graduation he joined the Wireless Communications lab of the DSPS Research and Development Center of Texas Instruments. He was a primary contributor in TI's effort in the IEEE wireless personal area networks standardization process. In 2001, Dr. Nafie joined Ellipsis Digital Systems, where he was responsible for the development of the Physical Layer (PHY) of 802.11b wireless LAN standard. He has also participated in designing and optimizing the PHY of the 802.11a standard. In 2002, Dr. Nafie co-founded SySDSoft, a company specializing in the design of wireless digital communication systems. He was the Chief Technology Officer of SySDSoft till December 2006. Dr. Nafie is currently a Professor at the Communications and Electronic Department of Cairo University, and at the Wireless Intelligent Networks Center of Nile University, for which he was a director from August 2009 to August 2012. His research interests are in the area of wireless and wireline digital communications and digital signal processing, and include equalization, LDPC coding, MIMO, relaying systems and caching among others. Dr. Nafie has published numerous Journal and Conference papers, and has nineteen issued US patents as well as few others pending. He has several patents in Europe and Japan as well.
\end{IEEEbiography}
\end{document}